\documentclass[11pt,a4paper]{article}
\usepackage{amssymb, amsmath,framed}
\usepackage[colorlinks]{hyperref}

\newcommand{\rem}[1]{}
\newcommand{\de}{{\rm d}}
\newcommand{\m}{{\mathsf{m}}}
\newcommand{\bq}{{\mathbf{x}}}
\newcommand{\bp}{{\mathbf{p}}}
\newcommand{\bm}{{\mathbf{m}}}
\newcommand{\bM}{{\mathbf{M}}}
\newcommand{\bA}{{\mathbf{A}}}
\newcommand{\bN}{{\mathbf{N}}}
\newcommand{\bK}{{\mathbf{K}}}
\newcommand{\bE}{{\mathbf{E}}}
\newcommand{\bB}{{\mathbf{B}}}
\newcommand{\bD}{{\mathbf{D}}}
\newcommand{\bu}{{\boldsymbol{u}}}
\newcommand{\bU}{{\boldsymbol{U}}}
\newcommand{\dvol}{{\de^3\bq\,\de^3\bp}}

\newtheorem{theorem}{Theorem}

\newtheorem{corollary}[theorem]{Corollary}
\newtheorem{proposition}[theorem]{Proposition}

\newtheorem{remark}{Remark}
\newenvironment{proof}[1][Proof]{\noindent\textbf{#1.} }{\ \rule{0.5em}{0.5em}}

\begin{document}

\title{Hamiltonian approach to hybrid plasma models\vspace{-.1cm}}
\author{Cesare Tronci\\
\it\small Section de Math\'ematiques, \'Ecole Polytechnique
F\'ed\'erale de Lausanne, Switzerland \\\vspace{-.7cm}}
\date{}

\maketitle

\begin{abstract}
The Hamiltonian structures of several hybrid kinetic-fluid models
are identified explicitly, upon considering collisionless Vlasov
dynamics for the hot particles interacting with a bulk fluid. After
presenting different pressure-coupling schemes for an ordinary fluid
interacting with a hot gas, the paper extends the treatment to
account for a fluid plasma interacting with an energetic ion
species. Both current-coupling and pressure-coupling MHD schemes are
treated extensively. In particular, pressure-coupling schemes are
shown to require a transport-like term in the Vlasov kinetic
equation, in order for the Hamiltonian structure to be preserved.
The last part of the paper is devoted to studying the more general
case of an energetic ion species interacting with a neutralizing
electron background (hybrid Hall-MHD). Circulation laws and Casimir
functionals are presented explicitly in each case.\vspace{-.1cm}
\end{abstract}

\tableofcontents

\section{Introduction}

\subsection{Hybrid plasma models}
While fluid models are widely successful in plasma physics, kinetic
effects have been shown to be relevant in many situations involving
coexistence of cold plasmas and energetic hot particles. The latter
play an important role in different contexts, ranging from fusion
research \cite{ParkEtAl} to astrophysical plasmas \cite{WiYiOmKaQu}.
The need for multiscale models that accommodate the statistical
kinetic effects of hot particles is a subject of current research,
particularly involving computer simulations \cite{PaBeFuTaStSu}.
These models are usually realized by following a hybrid philosophy
that couples ordinary fluid models to appropriate kinetic equations
governing the phase-space distribution of the energetic particle species.

The two main research directions involve either the coexistence of a
fluid MHD plasma with an energetic ion component \cite{ParkEtAl} or
the coexistence of an electron neutralizing background with an
energetic ion component \cite{WiYiOmKaQu,YiWiEtAl} (hybrid
Hall-MHD). The first direction splits into two possible approaches:
the current-coupling scheme \cite{ParkEtAl,BeDeCh,ToSaWaWaHo} and
the pressure-coupling scheme \cite{ParkEtAl,FuPark,Cheng,TaBrKi},
depending on how the fluid equation is coupled to the kinetic
equation for the hot particles.

In particular, the pressure-coupling scheme possesses some variants
that will be discussed in this paper. In all the variants appearing
in the literature, the kinetic-fluid coupling occurs solely as an
extra pressure term in the fluid momentum equation, while the hot
particles are affected by the cold background plasma only through
the electromagnetic field terms appearing in the kinetic equation.
The fluid plasma density is conventionally transported by the
background fluid and this transport generates the usual (barotropic)
pressure effects in the evolution of the fluid momentum. On the
other hand, the kinetic pressure corresponding to the hot particles
has never been considered as arising from fluid transport terms
appearing in the accompanying kinetic equation. The absence of these
fluid transport terms destroys certain relevant properties of
ordinary fluid models, such as the celebrated Kelvin circulation
theorem that is of paramount importance, for example, in geophysical
fluid dynamics. Also, while the momentum equation usually conserves
the prescribed energy, the absence of transport terms in the kinetic
equation breaks the corresponding Hamiltonian structure, which is of
 interest in various plasma models \cite{MaWeRaScSp,Morrison2,Morrison2005,BrTr}. However, even
when considering hypothetical transport terms, their inclusion in
the kinetic equation cannot be achieved by ad hoc arguments and
their insertion must follow from general principles. An interesting
work in this direction was carried out by {Fl\aa} \cite{Fla}, who
considered a linearized approach. Fl\aa's results follow from the
geometric properties underlying kinetic Vlasov-type equations and in
particular their Lie-group symmetries \cite{MaWe1,Weinstein}.

The pressure-coupling scheme has been adopted in many simulation
codes, probably because of its simplicity (compared to the
current-coupling scheme). On the other hand, this scheme assumes
that the averaged kinetic momentum of the hot particles is
negligible at all times. This assumption is not required by the
current-coupling scheme, which then possesses a wider range of
validity. The current-coupling scheme couples the fluid and kinetic
equations through the Lorentz force that is exerted by the hot
particles on the background fluid plasma. Although this scheme has
been used in \cite{BeDeCh,ToSaWaWaHo}, it remains less popular in
the plasma physics community. Nevertheless, this paper stands in
favor of this model, because of the strong assumption on the hot
particle momentum that is required by the pressure-coupling scheme.

All existing hybrid codes usually replace the phase-space kinetic
equation for the energetic component by particle dynamics. The
distribution function is then reconstructed in order to calculate
the macroscopic variables, such as pressure and momentum. In
particular, the hot particles are often advanced by appropriate
gyrokinetic equations \cite{Littlejohn,GrKaLi,BrHa}. Nevertheless,
this paper will describe the energetic ions in terms of the more
general Vlasov kinetic equation, whose Hamiltonian formulation is
easier then its gyrokinetic counterpart. The particle picture can
always be recovered by the $\delta$-like particle solution, which is
well known to preserve the Hamiltonian structure \cite{HoTr}.

\subsection{Hamiltonian methods in plasma physics}
The use of Hamiltonian methods in plasma physics goes back to the
early 80's, when many well known plasma models were found to possess
a Hamiltonian structure. The latter is composed of a Poisson bracket
and a Hamiltonian function(al), which unambiguously identifies the
total energy of the system. For example, celebrated Hamiltonian
structures are those for multi-fluid plasmas \cite{SpKa,Spencer},
for magnetohydrodynamics \cite{MorrisonGreene,HolmKupershmidt} and
for the Maxwell-Vlasov equations
\cite{Morrison2bis,MaWe1,Weinstein}. Another important discovery was
the Hamiltonian formulation of guiding-center motion, which led to
the modern theory of gyrokinetic equations
\cite{GrKaLi,Littlejohn,BrHa}. The development of Hamiltonian
techniques is still an active area of research, as shown by the
recent works on magnetic reconnection and gyrofluid models (see
\cite{Morrison2} and references therein).

The Hamiltonian formulation is particularly advantageous when it is
accompanied by a variational principle on the Lagrangian side, which
is the case for most plasma models \cite{Morrison2}. When this
happens, this is due to the fact that the Poisson bracket arises
naturally from the relabeling symmetry of the plasma
\cite{MaWeRaScSp}. The advantage of such a Hamiltonian structure
resides in several aspects, whose most celebrated example is the
possibility of carrying an exact nonlinear stability analysis
\cite{HoMaRaWe}. This analysis is made possible by the existence of
invariant functions (Casimirs) that arise from the Poisson bracket.
On the Lagrangian side, Noether's theorem produces conserved
quantities, thereby establishing the existence of explicit
circulation laws that are a fundamental feature of fluid models.
Moreover, the variational principle associated to the Hamiltonian
structure can be approached to apply asymptotic methods or
Lagrangian averaging (or Lie-transform) techniques. Some recent
reviews on these topics, with special emphasis on plasma physics,
are available in \cite{BrTr,Morrison2,Morrison2005}.

In certain cases, the Hamiltonian structure does not possess an
immediate variational principle formulation in terms of purely
Lagrangian variables. For example, although the Maxwell-Vlasov
system \cite{Morrison2bis,MaWe1,Weinstein} possesses several
variational formulations \cite{Pfirsch,PfMo1,PfMo2,YeMo,Low}, these
are  based either on
 Eulerian variables or on a mixture of Eulerian and
Lagrangian variables. On the other hand, purely Lagrangian variables
were used in \cite{CeMaHo}, upon suitably modifying Low's Lagrangian
\cite{Low}. The method in \cite{CeMaHo} is general enough to be
applicable to any Vlasov-type system. This paper relies on this
approach to produce the Lagrangian formulation of the Hamiltonian
plasma models that will be presented.

\subsection{Goal of the paper}
This paper aims to present the explicit Hamiltonian structure of
several hybrid kinetic-fluid models, for either ordinary barotropic
fluids, MHD and Hall MHD. In particular, this paper derives a whole
class of hybrid kinetic-fluid models by making use of well
established Hamiltonian Poisson bracket methods. Upon taking the
direct sum of ordinary Poisson bracket structures, the hybrid models
are derived by simply transforming the bracket appropriately, so
that the Hamiltonian structure is always carried along in a natural
fashion. The current-coupling MHD scheme is found to possess an
intrinsic Poisson bracket, which is presented explicitly. On the
other hand, pressure-coupling schemes is derived by neglecting the
momentum contribution of the hot particles directly in the form of
the Hamiltonian, rather than approximating the equations of motion
in their final form. This key step of approximating the Hamiltonian
instead of the equations leads to the preservation of all the usual
properties of fluid models, such as circulation theorems and even
helicity preservation. All the considerations in this paper restrict
to consider non-collisional Vlasov dynamics and barotropic fluid
flows. The more general case of adiabatic fluid flows accounting for
 specific entropy transport is a straightforward extension. After
presenting the most basic (pressure-coupling) hybrid models for
ordinary neutral fluids, the paper formulates hybrid models for
electromagnetic fluids. Hybrid versions of magnetohydrodynamics are
presented, together with its associated circulation theorems. The
last part of the paper is devoted to presenting the Hamiltonian
structure of the hybrid Hall-MHD scheme \cite{WiYiOmKaQu,YiWiEtAl}.

\begin{remark}[The role of Lie symmetries]
The whole paper makes use of well established Lie-symmetry
techniques to produce the various Hamiltonian structures. Indeed,
the use of Lie-symmetry concepts (e.g. momentum maps and group
actions) is of central importance to the treatment, although the
discussion ignores all technical points. Consequently, all the
resulting Poisson brackets possess an intrinsic geometric nature,
which goes back to the invariance properties underlying MHD \cite{MaWeRa} and the
Maxwell-Vlasov system \cite{MaWe1,MaWeRaScSp}. Thus, the emergence
of Lie-Poisson systems is an immediate consequence of the treatment.
Restricting to consider these symmetric Lie-Poisson structures is a
deliberate choice of the author and the possibility of other forms of Poisson
brackets is not taken under consideration.
\end{remark}

\section{Elementary hybrid models}
This section presents the most basic ideas underlying hybrid models.
We consider two simple situations: the hybrid description of a
single fluid and the hybrid model for an ideal fluid interacting
with an ensemble of hot particles.

\subsection{Hybrid formulation of ordinary fluid dynamics\label{hybrid-onefluid}}
The most fundamental hybrid model already arises in the well known
derivation of fluid dynamics from kinetic theory \cite{Liboff}. Upon
restricting to the non-collisional case, one considers a Vlasov
equation of the form
\begin{equation}\label{Vlasov1}
\frac{\partial f}{\partial t}+\bp\cdot\frac{\partial f}{\partial
\bq}+\mathbf{F}\cdot\frac{\partial f}{\partial \bp}=0.
\end{equation}
where $\bp$ is the momentum coordinate (the single-particle mass is
taken to be unitary for simplicity) and $f(\bq,\bp,t)$ is the Vlasov
distribution on phase-space. The collective force field
$\mathbf{F}=-\nabla\phi(f)$  arises from a potential $\phi$, which is usually a linear functional of $f$.
Equation \eqref{Vlasov1} conserves the total energy
\begin{equation}\label{Vlasov-Ham1}
H(f)=\frac12\int \!f
\left({\left|\bp\right|^2}+\phi(f)\right)\de^3\bq\,\de^3\bp
\end{equation}
and  possesses the Poisson structure
\cite{Morrison1,Morrison2,MaWeRaScSp,Morrison2005,Morrison3,Morrison1bis}
\begin{equation}\label{Vlasov-PB1}
\{F,G\}=\int f\left\{\frac{\delta F}{\delta f},\frac{\delta G}{\delta f}\right\}\de^3\bq\,\de^3\bp
\end{equation}
where $\{\cdot,\cdot\}$ denotes canonical Poisson bracket on phase space. Thus, equation \eqref{Vlasov1} becomes
\begin{equation}\label{Vlasov2}
\frac{\partial f}{\partial t}+\left\{f,\frac{\delta H}{\delta f}\right\}=0\,.
\end{equation}
The equations of fluid dynamics follow by closing the equations of the moments \cite{GiHoTr}
\begin{equation}\label{moments}
A_s(\bq,t)=\int\bp^{\otimes s}\, f(\bq,\bp,t)\,\de^3\bp
\end{equation}
where $\bp^{\otimes s}$ is the $s$-th tensor power on the momentum
coordinate. The moment hierarchy is conventionally closed to
consider the first two moments $(A_0,A_1)$ as the basic dynamical variables,
while the second-order moment dynamics is usually assumed to depend
on $(A_0,A_1)$. Evidently, $n=A_0$ is the particle density, while
$\bK=A_1$ is the averaged kinetic momentum, that is the fluid momentum. We
notice that the Hamiltonian \eqref{Vlasov-Ham1} can be rewritten in
terms of the first two moments as
\begin{equation}\label{hybrid-Ham1}
\bar{H}(n,\bK,f)=\frac1{2}\int\!
f\,\left|\bp-\frac{\bK}{n}\right|^2\dvol+\frac1{2}\int\frac{\left|\bK\right|^2}{n}\
\de^3\bq+\Phi(n)\,,
\end{equation}
where $\Phi$ denotes a potential energy functional depending
exclusively on the particle density $n$.

\begin{remark}[Mean and fluctuation terms in the Vlasov Hamiltonian] Notice that the above form
\eqref{hybrid-Ham1} of the Vlasov Hamiltonian \eqref{Vlasov-Ham1} has the advantage of
splitting explicitly the mean and fluctuation parts of the kinetic
energy, corresponding to the second and first term respectively.
When the fluctuation part is absent, the system reduces to a `cold
plasma' system, which usually arises from the simple moment closure
$f=n\,\delta(\bp-\bK/n)$. On the other hand, in certain situations
involving small amounts of energetic particles, the energy
contribution of the averaged momentum $\bK$ is often neglected.
Then, the second term in \eqref{hybrid-Ham1} is small and all the
kinetic energy is concentrated in the fluctuation part, which is the
first term. As we shall see, this point is of central interest in
hybrid models for plasma physics.
\end{remark}

In order to express the Vlasov equation \eqref{Vlasov1} in terms of
its macroscopic fluid quantities, one inserts \eqref{chainrule} the
chain rule formula
\begin{equation}\label{chainrule}
\frac{\delta F}{\delta f}=\frac{\delta \bar{F}}{\delta
n}+\bp\cdot\frac{\delta \bar{F}}{\delta \bK}+\frac{\delta
\bar{F}}{\delta f}
\end{equation}
in the Poisson bracket \eqref{Vlasov-PB1}, thereby obtaining the
 Poisson structure
\begin{align}\nonumber
\{\bar{F},\bar{G}\}(n,{\bK},f)=& \int{\bK}\cdot \left[\frac{\delta
\bar{F}}{\delta{\bK}},\frac{\delta \bar{G}}{\delta
{\bK}}\right]\de^3\bq - \int n\left(\frac{\delta \bar{F}}{\delta
\bK}\cdot\nabla\frac{\delta \bar{G}}{\delta n}-\frac{\delta
\bar{G}}{\delta \bK}\cdot \nabla\frac{\delta \bar{F}}{\delta
n}\right)\de^3\bq
\\\nonumber
&+ \int f\left(\left\{\frac{\delta \bar{F}}{\delta f},\frac{\delta
\bar{G}}{\delta n}+\bp\cdot\frac{\delta \bar{G}}{\delta
\bK}\right\}-\left\{\frac{\delta \bar{G}}{\delta f},\frac{\delta
\bar{F}}{\delta n}+\bp\cdot\frac{\delta \bar{F}}{\delta
\bK}\right\}\right)\dvol
\\
&+\int f\left\{\frac{\delta \bar{F}}{\delta f},\frac{\delta
\bar{G}}{\delta f}\right\}\dvol \,. \label{PB-basic-hybrid}
\end{align}
Here,
${[\boldsymbol{X},\boldsymbol{Y}]}=-{(\boldsymbol{X}\cdot\nabla)
\boldsymbol{Y}}+{(\boldsymbol{Y}\cdot\nabla) \boldsymbol{X}}$ is
minus the commutator on vector fields. Upon dropping the bar
notation for convenience, Hamilton's equations are written as:
\begin{align}
\label{hybrid1-Hdependent-density}
\frac{\partial n}{\partial
t}&+\operatorname{div}\!\left(n\,\frac{\delta H}{\delta
\bK}\right)=-\int\left\{f,\frac{\delta H}{\delta f}\right\}\de^3\bp
\\
\label{hybrid1-Hdependent-momentum}
\frac{\partial \bK}{\partial
t}&+\operatorname{div}\!\left(\frac{\delta H}{\delta
\bK}\,\bK\right) + \left(\nabla\frac{\delta H}{\delta
\bK}\right)\cdot\bK=-n\nabla\frac{\delta H}{\delta n} -
\int\bp\left\{f,\frac{\delta H}{\delta f}\right\}\de^3\bp
\\
\frac{\partial f}{\partial t}&+\left\{f,\frac{\delta H}{\delta
n}+\bp\cdot\frac{\delta H}{\delta
\bK}\right\}=-\left\{f,\frac{\delta H}{\delta f}\right\}
\label{hybrid1-Hdependent-Vlasov}
\end{align}
These must be accompanied by the constraints
\[
\bK-\int\!\bp\,f\,\de\bp=0\,,\qquad\quad n-\int \!f\,\de\bp=0\,,
\]
which enforce the moments $(\bK,n)$ to depend on the distribution
function $f$ at all times. Indeed, although the hydrodynamical quantities
$(\bK,n)$ are mutually independent, they both depend on the statistical distribution $f$. (In more geometric terms, the above constraints
identify a zero-level set of a momentum map, known as
plasma-to-fluid map \cite{MaWeRaScSp}).

At this point, the Hamiltonian \eqref{hybrid-Ham1} yields
\[
\frac{\delta \bar{H}}{\delta
f}=\frac12\left|\bp-\frac{\bK}{n}\right|^2\,,\qquad \frac{\delta
\bar{H}}{\delta \bK}=\frac{\bK}{n}\,,\qquad \frac{\delta\bar{H}}{\delta
n}=-\frac12\left|\frac{\bK}{n}\right|^2+\frac{\delta \Phi}{\delta n}
\,,
\]
thereby producing the pressure term as follows:
\begin{proposition}\label{pressure-lemma} With the notation above, the following relation holds
\[
\frac12\int\!\bp\left\{f,\left|\bp-\boldsymbol{V}\right|^2\right\}\de^3\bp
=
\operatorname{div}\int\!\left(\bp-\boldsymbol{V}\right)^{\otimes\,2}
f\ \de^3\bp
\,,
\]
where we have defined the velocity $\boldsymbol{V}=\bK/n$.
\end{proposition}
\textbf{Proof.} The proof is given by direct verification as
follows:
\begin{align*}
\frac12\int\!\bp\left\{f,\left|\bp-\boldsymbol{V}\right|^2\right\}\de^3\bp=&
\int\!\bp\left\{f,\left(\frac12\,|\bp|^2-\bp\cdot\boldsymbol{V}+\frac12\,|\boldsymbol{V}|^2\right)\right\}\de^3\bp
\\
=&\operatorname{div}\!\int\!\bp\bp\,f\,\de^3\bp -
(\boldsymbol{V}\cdot\nabla)\bK-\nabla\boldsymbol{V}\cdot\bK-(\operatorname{div}\boldsymbol{V})\bK
+ \frac12\,n\nabla|\boldsymbol{V}|^2
\\
=& \operatorname{div}\!\int\!\bp\bp\,f\,\de^3\bp
-(\boldsymbol{V}\cdot\nabla)\bK-(\operatorname{div}\boldsymbol{V})\bK
\\
=&
\operatorname{div}\!\left(n\,\boldsymbol{V}\boldsymbol{V}+\int\!\!\left(\bp-\boldsymbol{V}\right)\left(\bp-\boldsymbol{V}\right)
f\, \de^3\bp\right)
-(\boldsymbol{V}\cdot\nabla)\bK-(\operatorname{div}\boldsymbol{V})\bK
\\
=&
\operatorname{div}\!\int\!\!\left(\bp-\boldsymbol{V}\right)\left(\bp-\boldsymbol{V}\right)
f\, \de^3\bp\,. \qquad \blacksquare
\end{align*}
Therefore, the hybrid equations of motion read as
\begin{align}
\frac{\partial n}{\partial
t}&+\operatorname{div}\!\left(n\,\boldsymbol{V}\right)=0
\\
\frac{\partial \boldsymbol{V}}{\partial
t}&+(\boldsymbol{V}\cdot\nabla)\boldsymbol{V}=-\nabla \phi
-\frac1{n}\operatorname{div}\!\int(\bp-\boldsymbol{V})^{\otimes\,2}\,f\,\de^3\bp
\\
\frac{\partial f}{\partial t}&=-\bp\cdot\frac{\partial
f}{\partial\bq}+\frac{\partial \phi}{\partial\bq}\cdot\frac{\partial
f}{\partial\bp} \,,
\end{align}
where we have defined the potential energy density
$\phi=\delta\Phi/\delta n$ and $\boldsymbol{V}=\bK/n$.

At this point, one usually proceeds by  invoking thermodynamic
principles leading to an equation of state. This simplifies the
pressure tensor on the right hand side of the velocity equation and
allows to discard the kinetic effects contained in the Vlasov
equation for $f$. This process leads to ordinary fluid equations for
inviscid barotropic fluids. However, when kinetic effects cannot be
neglected, the Vlasov kinetic equation must be retained in the above
system, which then constitutes the hybrid formulation of
ordinary fluid dynamics.

We notice that the terms on the left-hand side of the kinetic
equation in \eqref{hybrid1-Hdependent-Vlasov} play a key role in the
Hamiltonian structure of the system by coupling fluid and kinetic
terms in the Poisson bracket. Even more importantly, the term
$\bp\cdot\delta H/\delta \bK$ automatically generates the pressure
term in the velocity equation: this presence of the pressure tensor in the velocity equation
is not allowed when the term $\bp\cdot\delta H/\delta \bK$ is
absent in equation \eqref{hybrid1-Hdependent-Vlasov}. Moreover, the same term ensures the validity of the Kelvin
circulation theorem
\begin{align}
\frac{\de}{\de
t}\oint_{\gamma_t(\boldsymbol{V})}\boldsymbol{V}\cdot\de\bq=-\oint_{\gamma_t(\boldsymbol{V})}\left(\frac1{\rho}\,\operatorname{div}\!\int(\bp-\boldsymbol{V})^{\otimes\,2}\,f\,\de^3\bp\right)\cdot\de\bq
\,,
\end{align}
where $\gamma_t(\boldsymbol{V})$ is an arbitrary loop moving with
 velocity $\boldsymbol{V}$ and the right-hand side evidently
arises from the kinetic energy associated to the fluctuation
velocity $\bp-\boldsymbol{V}$. When this fluctuation velocity is
absent in the Hamiltonian \eqref{hybrid-Ham1} the right-hand side
above vanishes, thereby returning the usual Kelvin's theorem
$\oint_{\gamma_t}\boldsymbol{V}=\operatorname{const.}$

\subsection{Hamiltonian hybrid kinetic-fluid systems} This section considers the
interaction of an ideal compressible fluid interacting with an
ensemble of energetic particles. This basic example is of
fundamental importance because it embodies all the main properties
that will be used later to derive hybrid pressure-coupling plasma models.

In order to derive the hybrid system, we begin by considering the
following equations for the fluid density $\rho$ and velocity
$\boldsymbol{u}$ and for the Vlasov distribution $f$:
\begin{align}\label{Eul-Vlas-unt1bis}
\frac{\partial \rho}{\partial
t}&+\operatorname{div}(\rho\,\boldsymbol{u})=0
\\
\label{Eul-Vlas-unt1} \frac{\partial \boldsymbol{u}}{\partial
t}&+\boldsymbol{u}\cdot\nabla\boldsymbol{u}=-\frac1\rho\nabla{\sf
p}-\nabla\frac{\delta \Phi}{\delta \rho}
\\
\frac{\partial f}{\partial t}&+\bp\cdot\frac{\partial
f}{\partial\bq}-\nabla\frac{\delta \Phi}{\delta
f}\cdot\frac{\partial f}{\partial\bp}=0 \,. \label{Eul-Vlas-unt}
\end{align}
Here $\mathsf{p}$ denotes the scalar fluid pressure while the total
potential energy $\Phi(\rho,f)$ is a functional of both the fluid
mass density and the Vlasov distribution. For instance, the case of
electrostatic interactions involves a potential energy functional of
the form \cite{Morrison1,MaWeRaScSp}
\[
\Phi=-\frac12\iint\!\left(q_c\,\rho(\bq)+q_h\!
\int\!f(\bq,\bp)\,\de^3\bp\right)\frac1{\,\left|\bq-\bq'\right|^2}\left(q_c\,\rho(\bq')+q_h\!
\int\!f(\bq',\bp')\,\de^3\bp'\right)\de^3\bq\,\de^3\bq'\,,
\]
where the charge labels $c$ and $h$ denote the fluid (cold)
component and the energetic (hot) component respectively (again, we
consider unitary particle masses for simplicity). Then, the
potential energy represents the only coupling term of the above
equations, which otherwise are completely decoupled.

Since the Hamiltonian structures of ideal compressible fluids and
kinetic Vlasov equations are well known \cite{MaWeRaScSp}, it is
easy to see that the above system possesses a Hamiltonian
formulation with the direct sum Poisson bracket
\begin{equation}
\{F,G\}=\int\!{\bm}\cdot \left[\frac{\delta
F}{\delta{\bm}},\frac{\delta G}{\delta {\bm}}\right]\de^3\bq- \int\!
\rho\left(\frac{\delta{F}}{\delta \bm}\cdot\nabla\frac{\delta
{G}}{\delta \rho}-\frac{\delta{G}}{\delta \bm}\cdot
\nabla\frac{\delta{F}}{\delta \rho}\right)\de^3\bq +\int\!
f\left\{\frac{\delta F}{\delta f},\frac{\delta G}{\delta f}\right\}
\dvol \label{PB-Eul_Vlas-unt}
\end{equation}
and the Hamiltonian
\begin{equation}
H(\bm,\rho,f)=\frac1{2}\int\!\frac{\,\left|\bm\right|^2}{\rho}\
\de^3\bq+\frac1{2}\int\! f\left|\bp\right|^2\dvol
+\Phi(\rho,f)+\int\!\rho\ \mathcal{U}(\rho)\, \de^3\bq\,,
\end{equation}
which are both written in terms of the fluid momentum
$\bm=\rho\,\boldsymbol{u}$. Notice that we have introduced the fluid
internal energy $\mathcal{U}$ such that
$\mathsf{p}=\rho^2\,\de\mathcal{U}/\de\rho$.

At this point, in the search for a hybrid kinetic-fluid model, one
defines the total momentum
\[
\bM=\bm+\int\bp\, f\, \de^3\bp
\]
and expresses the equations of motion in terms of the new variable. This can be done by replacing the chain-rule formulas
\[
\frac{\delta F}{\delta\bm}=\frac{\delta \bar{F}}{\delta
\bM} \,,\qquad\frac{\delta F}{\delta\rho}=\frac{\delta \bar{F}}{\delta
\rho} \,,\qquad   \frac{\delta F}{\delta f}=\bp\cdot\frac{\delta \bar{F}}{\delta
\bM}+\frac{\delta \bar{F}}{\delta f}\,.
\]
in the Poisson bracket \eqref{PB-Eul_Vlas-unt}, thereby obtaining the structure
\begin{align}\nonumber
\{\bar{F},\bar{G}\}({\bM},\rho,&f)= \int{\bM}\cdot \left[\frac{\delta
\bar{F}}{\delta{\bM}},\frac{\delta \bar{G}}{\delta {\bM}}\right]\de^3\bq-
\int\! \rho\left(\frac{\delta \bar{F}}{\delta \bM}\cdot\nabla\frac{\delta
\bar{G}}{\delta \rho}-\frac{\delta \bar{G}}{\delta \bM}\cdot \nabla\frac{\delta
\bar{F}}{\delta \rho}\right)\de^3\bq
\\
&
+\int f\left\{\frac{\delta \bar{F}}{\delta f},\frac{\delta \bar{G}}{\delta
f}\right\}\dvol
+ \int f\left(\left\{\frac{\delta \bar{F}}{\delta f},\bp\cdot\frac{\delta \bar{G}}{\delta
\bM}\right\}-\left\{\frac{\delta \bar{G}}{\delta f},\bp\cdot\frac{\delta \bar{F}}{\delta
\bM}\right\}\right)\dvol
\,.
\label{PB-pressurescheme1}
\end{align}
Then, upon denoting $\bK=\int\!\bp\,f\,\de^3\bp$, the total energy expression
\begin{equation}\label{Hamiltonian-bar}
\bar{H}(\bM,\rho,f)=\frac1{2}\int\!\frac{\,\left|\bM-\bK\right|^2}{\rho}\
\de^3\bq+\frac1{2}\int\! f\left|\bp\right|^2\dvol
+\Phi(\rho,f)+\int\!\rho\ \mathcal{U}(\rho)\, \de^3\bq\,,
\end{equation}
yields the equations of motion
\begin{align}
\frac{\partial \rho}{\partial
t}&+\operatorname{div}(\rho\,\boldsymbol{u})=0
\\
\frac{\partial \boldsymbol{U}}{\partial
t}&+\boldsymbol{u}\cdot\nabla\boldsymbol{u}=-\frac1\rho\operatorname{div}\!\int\!\bp\bp\,f\,\de^3\bp-\frac1\rho\nabla{\sf
p}-\nabla\frac{\delta \Phi}{\delta \rho}
\\
\frac{\partial f}{\partial t}&+\bp\cdot\frac{\partial
f}{\partial\bq}-\frac{\partial
}{\partial\bq}\frac{\delta \Phi}{\delta
f}\cdot\frac{\partial f}{\partial\bp}=0 \,.
\end{align}
where $\boldsymbol{U}=\bM/\rho$ and
$\boldsymbol{u}=\bm/\rho=(\bM-\bK)/\rho$. These last equations of
motion are totally equivalent to the system
\eqref{Eul-Vlas-unt1bis}-\eqref{Eul-Vlas-unt}, that is they yet
involve no approximations.
\subsubsection{First pressure-coupling scheme}
As we remarked in the preceding section, one normally considers
small amounts of energetic particles, so that the total energy
contribution of the averaged momentum is often neglected. In the
plasma physics literature, this assumption is made by simply
replacing the convective term
$\boldsymbol{u}\cdot\nabla\boldsymbol{u}$ by
$\boldsymbol{U}\cdot\nabla\boldsymbol{U}$ in the velocity equation.
On the other hand, this simple step inexorably destroys the
Hamiltonian structure of the equations of motion, which then lack an
energy balance equation and therefore need substantial
modifications. A simple possibility of preserving the Hamiltonian
structure is provided by the assumption
$\bM=\bm+n\,\boldsymbol{V}\simeq\bm$ in the Hamiltonian
\eqref{Hamiltonian-bar}. Then, the total energy
\begin{equation}
\bar{H}(\bM,\rho,f)=\frac1{2}\int\!\frac{\,\left|\bM\right|^2}{\rho}\
\de^3\bq+\frac1{2}\int\! f\left|\bp\right|^2\dvol
+\Phi(\rho,f)+\int\!\rho\ \mathcal{U}(\rho)\, \de^3\bq\,,
\end{equation}
yields the equations of motion
\begin{align}
\frac{\partial \rho}{\partial
t}&+\operatorname{div}(\rho\,\boldsymbol{U})=0
\\
\frac{\partial \boldsymbol{U}}{\partial
t}&+\boldsymbol{U}\cdot\nabla\boldsymbol{U}=-\frac1\rho\operatorname{div}\!\int\!\bp\bp\,f\,\de^3\bp-\frac1\rho\nabla{\sf
p}-\nabla\frac{\delta \Phi}{\delta \rho}
\\
\frac{\partial f}{\partial t}&+\left\{f,{\bp}\cdot\boldsymbol{U}\right\}+\bp\cdot\frac{\partial
f}{\partial\bq}-\frac{\partial
}{\partial\bq}\frac{\delta \Phi}{\delta
f}\cdot\frac{\partial f}{\partial\bp}=0 \,.
\end{align}
While the velocity equation is exactly the same as that of some
hybrid schemes frequently used in plasma physics
\cite{FuPark,ParkEtAl,Cheng} (up to Lorentz force terms that will be
included in the following sections), the Vlasov equation carries the
extra term
\[
\left\{f,{\bp}\cdot\boldsymbol{U}\right\}=\boldsymbol{U}\cdot\frac{\partial
f}{\partial\bq}-\frac{\partial f}{\partial\bp}\cdot\frac{\partial\boldsymbol{U}
}{\partial\bq}\cdot\bp
\]
which evidently arises from the velocity shift
$\boldsymbol{U}=\boldsymbol{u}+\bK/\rho$.

\begin{remark}[Transport of the Vlasov distribution]
We emphasize that the term $\left\{f,\bp\cdot\boldsymbol{U}\right\}$
is required by the Hamiltonian structure of the original equations
and cannot arise from direct approximations on the equations of
motion \eqref{Eul-Vlas-unt1bis}-\eqref{Eul-Vlas-unt}. The same term
was found by Fl{\aa} \cite{Fla}, who used more sophisticated
techniques. We notice that the term $\boldsymbol{U}\cdot\nabla f$
leads to the interpretation of a Vlasov equation that is
`transported' along the macroscopic velocity $\boldsymbol{U}$. This
transport property generates the circulation force term
$\partial_\bp f\cdot\nabla\boldsymbol{U}\cdot\bp$ that is related to
the change of reference $\boldsymbol{u}\to\boldsymbol{U}$. The
kinetic-fluid interaction term
$\left\{f,\bp\cdot\boldsymbol{U}\right\}$ is the main novelty of
this model.
\end{remark}

 The Vlasov equation may be rewritten as $\partial_tf+\big\{f,1/2\big(\left|\bp+\boldsymbol{U}\right|^2-\left|\boldsymbol{U}\right|^2\big)+\delta_{\!f}\Phi\big\}=0$. Notice that the equation ${\partial_t n +\operatorname{div}(n\,\boldsymbol{U})}=-\operatorname{div}\bK$ for the energetic particle density $n=\int\!f\,\de^3\bp$ involves the continuous source term $\operatorname{div}\bK$. Thus, a slowly varying density $n$ enforces the averaged momentum
$\bK$ to possess negligible divergence.

\subsubsection{Second pressure-coupling scheme}
So far, we only made the simple assumption $\bM\simeq\bm$ in the expression of the total energy. This was justified by the fact that the energetic particles exist only in small amounts, so that the averaged momentum $\bK$ is assumed not to contribute to the total energy. However, as we have remarked in section \ref{hybrid-onefluid}, under this assumption the kinetic energy of the energetic component may be modified to neglect its average value $|\bK|^2/(2n)$. Then, the total energy becomes
\begin{equation}\label{Hamiltonian-pressurecoupling2}
\bar{H}(\bM,\rho,f)=\frac1{2}\int\!\frac{\,\left|\bM\right|^2}{\rho}\
\de^3\bq+\frac1{2}\int\! f\left|\bp-\frac{\bK}{n}\right|^2\dvol
+\Phi(\rho,f)+\int\!\rho\ \mathcal{U}(\rho)\, \de^3\bq\,,
\end{equation}
which in turn produces the equations
\begin{align}
\frac{\partial \rho}{\partial
t}&+\operatorname{div}(\rho\,\boldsymbol{U})=0
\\
\frac{\partial \boldsymbol{U}}{\partial
t}&+\boldsymbol{U}\cdot\nabla\boldsymbol{U}=-\frac1\rho\operatorname{div}\!\int\!\left(\bp-\boldsymbol{V}\right)\left(\bp-\boldsymbol{V}\right)f\,\de^3\bp-\frac1\rho\nabla{\sf
p}-\nabla\frac{\delta \Phi}{\delta \rho}
\\
\frac{\partial f}{\partial t}&+\left\{f,{\bp}\cdot\big(\boldsymbol{U}-\boldsymbol{V}\big)+\frac12|\boldsymbol{V}|^2\right\}+\bp\cdot\frac{\partial
f}{\partial\bq}-\frac{\partial
}{\partial\bq}\frac{\delta \Phi}{\delta
f}\cdot\frac{\partial f}{\partial\bp}=0 \,.
\end{align}
The same form of the above velocity equation also appears in certain
hybrid schemes \cite{TaBrKi}, although the above Vlasov equation
differs substantially from those models because the present system
accounts for energy balance, contrarily to the pressure-coupling
schemes appearing in the physics literature. We notice that the
above Vlasov equation can be written in the form $\partial_t
f+\big\{f,1/2\big(\left|\bp-\boldsymbol{V}\right|^2+\left|\bp+\boldsymbol{U}\right|^2-\left|
\bp\right|^2-\left|\boldsymbol{U}\right|^2\big)+\delta_{\!f}\Phi\big\}=0$,
which emphasizes the various kinetic energy contributions. A natural
consequence of this kinetic equation is that its zero-th moment $n$
satisfies the simple advection relation
$\partial_tn+\operatorname{div}(n\,\boldsymbol{U})=0$, so that the
total density $D=\rho+n$ satisfies the equation $\partial_t
D+\operatorname{div}(D\,\boldsymbol{U})=0$, as it appears in \cite{KimEtAl}.

\subsubsection{Third pressure-coupling scheme\label{thirdPressureScheme}}
The advection equation for the total density $D$ arising in the last case, suggests looking at the equations of motion in terms of this density variable. To this purpose, one again replaces the chain-rule formulas
\[
\frac{\delta F}{\delta\rho}=\frac{\delta \bar{F}}{\delta
D}\,,\qquad\frac{\delta F}{\delta\bm}=\frac{\delta \bar{F}}{\delta
\bM} \,,\qquad  \frac{\delta F}{\delta f}=\frac{\delta
\bar{F}}{\delta D}+\bp\cdot\frac{\delta \bar{F}}{\delta
\bM}+\frac{\delta \bar{F}}{\delta f}\,.
\]
in the Poisson bracket \eqref{PB-Eul_Vlas-unt}, thereby obtaining the
following Poisson bracket
\begin{align}\nonumber
\{\bar{F},\bar{G}\}({\bM},D,f)=& \int{\bM}\cdot \left[\frac{\delta
\bar{F}}{\delta{\bM}},\frac{\delta \bar{G}}{\delta {\bM}}\right]\de^3\bq-
\int\! D\left(\frac{\delta \bar{F}}{\delta \bM}\cdot\nabla\frac{\delta
\bar{G}}{\delta D}-\frac{\delta \bar{G}}{\delta \bM}\cdot \nabla\frac{\delta
\bar{F}}{\delta D}\right)\de^3\bq
\\\nonumber
&+ \int f\left(\left\{\frac{\delta \bar{F}}{\delta f},\frac{\delta
\bar{G}}{\delta D}+\bp\cdot\frac{\delta \bar{G}}{\delta
\bM}\right\}-\left\{\frac{\delta \bar{G}}{\delta f},\frac{\delta
\bar{F}}{\delta D}+\bp\cdot\frac{\delta \bar{F}}{\delta
\bM}\right\}\right)\dvol
\\\nonumber
&+\int f\left\{\frac{\delta \bar{F}}{\delta f},\frac{\delta \bar{G}}{\delta
f}\right\}\dvol
 \,, \label{PB-Eul_Vlas-ent}
\end{align}
which is identical to \eqref{PB-basic-hybrid}. Under the assumption of a rarefied energetic gas, one has $D\simeq n$ and each of the Hamiltonians used above yields a slightly different system. For example, upon replacing $\rho$ by $D$, the energy in \eqref{Hamiltonian-pressurecoupling2} produces the equations
\begin{align}
\frac{\partial D}{\partial
t}&+\operatorname{div}(D\,\boldsymbol{U})=0
\\
\frac{\partial \boldsymbol{U}}{\partial
t}&+\boldsymbol{U}\cdot\nabla\boldsymbol{U}=-\frac1D\operatorname{div}\!\int\!\left(\bp-\boldsymbol{V}\right)\left(\bp-\boldsymbol{V}\right)f\,\de^3\bp-\frac1\rho\nabla{\sf
p}-\nabla\frac{\delta \Phi}{\delta D}
\\
\frac{\partial f}{\partial t}&+\left\{f,{\bp}\cdot\big(\boldsymbol{U}-\boldsymbol{V}\big)+\frac12\left(|\boldsymbol{V}|^2-|\boldsymbol{U}|^2\right)\right\}+\bp\cdot\frac{\partial
f}{\partial\bq}-\frac{\partial
}{\partial\bq}\frac{\delta \Phi}{\delta
f}\cdot\frac{\partial f}{\partial\bp}=0 \,,
\end{align}
which involve no substantial modification with respect to the previous hybrid schemes.

\subsection{Conservation laws}\label{Hybrid-Euler-Discussion}
The Hamiltonian structure of fluid models has the main advantage of preserving most of the geometric features possessed by ordinary ideal fluids. One of the most significant features is the Kelvin circulation theorem
\begin{equation}\label{KN1}
\frac{\de}{\de t}\oint_{\gamma_t(\boldsymbol{u})}\!
\boldsymbol{u}\cdot\de\bq=0
\end{equation}
 that holds for the
velocity equation \eqref{Eul-Vlas-unt1}. Here the curve
$\gamma_t(\boldsymbol{u})$ is a loop moving with velocity
$\boldsymbol{u}$. The hybrid schemes available in the physics
literature often lack this fundamental feature, which is instead
available for the hybrid models presented above. Indeed, the
momentum shift $\bM=\bm+\int\!\bp\,f\,\de^3\bp$ together with its
resulting Poisson bracket \eqref{PB-pressurescheme1} guarantee the
Kelvin circulation theorem
\begin{equation}\label{KN2}
\frac{\de}{\de t}\oint_{\gamma_t(\boldsymbol{U})}
\left(\boldsymbol{U}-\frac1\rho\int\!\bp\,f\,\de^3\bp\right)\cdot\de\bq=0
\,.
\end{equation}
It is important to observe that the circulation law \eqref{KN2} is
\emph{different} from \eqref{KN1}, since the loop $\gamma_t$ moves
with velocity $\boldsymbol{u}$ in \eqref{KN1} and with velocity
$\boldsymbol{U}=\boldsymbol{u}+\rho^{-1}\!\int\!\bp\, f\,\de^3\bp$
in \eqref{KN2}. This is due to the fact that modifying the fluid
kinetic energy in the Hamiltonian \eqref{Hamiltonian-bar} yields a
different fluid velocity.

\begin{remark}[Relation to turbulence models]
Notice that exactly the same phenomenon occurs in certain
regularization models for fluid turbulence \cite{ChFr}, such as the
Navier-Stokes-$\alpha$ model \cite{FoHoTi}. These models have also
been called ``Kelvin-filtered'' turbulence models. Another property
of the preceding pressure-coupling schemes that is shared with
Kelvin-filtered models is that the fluid vorticity
$\boldsymbol\omega=\nabla\times\boldsymbol{u}$ is transported along
the modified velocity $\bU$, that is
${\partial_{t\,}\boldsymbol\omega+\nabla\times\left(\boldsymbol{U}\times\boldsymbol\omega\right)=0}$.
This equation evidently possesses vortex filament solutions that
project to point vortices in two dimensions. The appearance of this
close analogy between pressure-coupling hybrid schemes and fluid
turbulence models looks promising for further investigation,
especially in relation to the energy splitting appearing in the
Hamiltonian \eqref{hybrid-Ham1}.
\end{remark}

The circulation conservation \eqref{KN2} is valid in any of the
pressure schemes presented previously. In the third
pressure-coupling scheme, the Kelvin circulation theorem becomes
\[
\frac{\de}{\de t}\oint_{\gamma_t(\bU)}
\left(\boldsymbol{U}-\left(D-\int\!f\,\de^3\bp\right)^{\!-1}\!\int\!\bp\,f\,\de^3\bp\right)\cdot\de\bq=0
\,.
\]

Moreover, the previous hybrid schemes possess the helicity conservation relation
\[
\frac{\de}{\de
t}\int\left(\boldsymbol{U}-\frac1\rho\int\!\bp\,f\,\de^3\bp\right)\cdot\nabla\times\left(\boldsymbol{U}-\frac1\rho\int\!\bp\,f\,\de^3\bp\right)\de^3\bq=0 \,.
\]
Notice the above conserved quantity is a Casimir of the Poisson
bracket \eqref{PB-pressurescheme1} and in two dimensions, this
Casimir can be used to approach the problem of nonlinear stability,
following the reference \cite{HoMaRaWe}. Another Casimir of the
Poisson bracket \eqref{PB-pressurescheme1} is given by
$C=\int\Phi(f)\,\de^3\bq\,\de^3\bp$, for an arbitrary function
$\Phi$. This is the Casimir typically associated to Vlasov
dynamics and it includes the total mass $\int\!f\,
\de^3\bq\,\de^3\bp$. The search for other Casimir functionals is
left for future investigation.

\section{Hamiltonian hybrid MHD models}
The purpose of this section is to formulate a kinetic-multifluid model that leads to the formulation of hybrid MHD schemes. As we shall see, the pressure-coupling schemes can be obtained by repeating systematically all the steps outlined in the previous sections, upon including electromagnetic fields appropriately. On the other hand there is another relevant hybrid MHD scheme, which involves a current coupling. This section derives both current and pressure coupling schemes and provides their Hamiltonian structures. At the end of this section, the two schemes are compared in terms of their energy-balance properties.

\subsection{Kinetic-multifluid system}\label{kinetic-mf-system}
The equations of motion for a multifluid plasma
in the presence of an energetic component read as
\begin{align}\label{multi-fluid-momentum}
&\rho_s\frac{\partial\bu_s}{\partial t}+\rho_s\left(\bu_s\cdot\nabla\right)\bu_s
=
a_s\rho_s\left(\bE+\bu_s\times\bB\right)-\nabla\mathsf{p}_s
\\
&
\frac{\partial\rho_s}{\partial t}+\operatorname{div}\left(\rho_s\bu_s\right)=0
\\\label{multi-fluid-Vlasov}
&
\frac{\partial f}{\partial t}+\frac{\bp}{m_h}\cdot\frac{\partial f}{\partial \bq}+q_h\left(\bE+\frac{\bp}{m_h}\times\bB\right)\cdot\frac{\partial f}{\partial \bp}=0
\\
&
\mu_0\epsilon_0\frac{\partial\bE}{\partial t}=\nabla\times\bB-\mu_0\sum_s a_s\rho_s\bu_s-\mu_0\,a_h\int\!\bp\,f\,\de^3\bp
\\
&
\frac{\partial\bB}{\partial t}=-\nabla\times\bE
\\
&
\epsilon_0
\operatorname{div}\bE=\sum_s a_s\rho_s+q_h\int f\,\de^3\bp
\,,\qquad
\operatorname{div}\bB=0
\end{align}
where $a_s=q_s/m_s$ is the charge-to-mass ratio of the fluid species
$s$, while $\rho_s$ and $\bu_s$ are its mass density and velocity.
Notice that in the above system we have restored the hot particle
mass $m_h$. When the energetic component is absent, the Hamiltonian
structure of this system was discovered by Spencer
\cite{Spencer,SpKa} and analyzed further in \cite{MaWeRaScSp}.
Combining this Hamiltonian structure with that of the Maxwell-Vlasov
system \cite{MaWe1,Weinstein,Morrison2bis} yields the following Poisson bracket
for the kinetic-multifluid system
\begin{align}\nonumber
\{F,G\}=& \sum_s\int{\bm}_s\cdot \left[\frac{\delta
F}{\delta{\bm}_s},\frac{\delta G}{\delta
{\bm}_s}\right]\de^3\bq
+
\int f\left\{\frac{\delta
F}{\delta f},\frac{\delta G}{\delta f}\right\}\dvol
\\\nonumber
&
-
\sum_s\int \rho_s \left(\frac{\delta F}{\delta \bm_s}\cdot\nabla\frac{\delta G}{\delta \rho_s}-\frac{\delta G}{\delta \bm_s}\cdot \nabla\frac{\delta F}{\delta \rho_s}\right)\de^3\bq
\\\nonumber
&
+
\sum_s\int\! a_s\rho_s \left(
\frac{\delta
F}{\delta \bm_s}\cdot\frac{\delta
G}{\delta \bD}-\frac{\delta
G}{\delta \bm_s}\cdot\frac{\delta
F}{\delta \bD}+\bB\cdot\frac{\delta
F}{\delta \bm_s}\times\frac{\delta
G}{\delta \bm_s}\right)\de^3\bq
\\\nonumber
&
+
q_h
\int f \left(\frac{\partial}{\partial \bp}\frac{\delta
F}{\delta f}\cdot\frac{\delta
G}{\delta \bD}-\frac{\partial}{\partial \bp}\frac{\delta
G}{\delta f}\cdot\frac{\delta
F}{\delta \bD}+\bB\cdot\left(\frac{\partial}{\partial \bp}\frac{\delta
F}{\delta f}\times\frac{\partial}{\partial \bp}\frac{\delta
G}{\delta f}\right)\right)\dvol
\\
&+
\int\left(\frac{\delta
F}{\delta \bD}\cdot\operatorname{curl}\frac{\delta
G}{\delta \bB}-\frac{\delta
G}{\delta  \bD}\cdot\operatorname{curl}\frac{\delta
F}{\delta \bB}\right)\de^3\bq
\,,
\label{PB-kinetic-multifluid}
\end{align}
where evidently $\bm_s=\rho_s\bu_s$. To complete the Hamiltonian structure of the kinetic-multifluid model, one writes the Hamiltonian
\begin{multline}
H(\bm_s,\rho_s,f,\bD,\bB)=\frac1{2}\sum_s\int \frac{\left|\bm_s\right|^2}{\rho_s}\,\de^3\bq+\frac1{2\m_h}\int f \left|\bp\right|^2\dvol+\sum_s\int\!\rho_s\, \mathcal{U}(\rho_s)\,\de^3\bq
\\
+\frac12\left(\frac1{\epsilon_0}\int \left|\bD\right|^2\de^3\bq+\frac1{\mu_0}\int \left|\bB\right|^2\de^3\bq\right)
\,,
\end{multline}
where the electric induction $\bD=\epsilon_0\bE$ is (minus) the conjugate
variable to the vector potential $\mathbf{A}$, as explained in
\cite{Holm1,Holm2}.

\begin{remark}[Extending to several energetic components]
The above kinetic-multifluid system can be easily extended to the
case of several energetic components, each described by its own
Vlasov equation. Then, each of these Vlasov equations can be
rewritten in terms of the first two  moments $(n,\bK)$ in order to
isolate kinetic pressure terms. A similar approach was followed in
\cite{ChengJohnson}.
\end{remark}

\subsection{Current-coupling hybrid MHD scheme}
In usual situations, one is interested in one-fluid models. Thus, it is customary to specialize to the two-fluid system and to neglect the inertia of one species (electrons). This last approximation is equivalent to taking the limit $m_2\to 0$ for the second species in the total fluid momentum equation. Under this assumption, the sum of the equations \eqref{multi-fluid-momentum} for $s=1,2$ produces the equation
\begin{equation}\label{TotMom}
\rho_1\frac{\partial\bu_1}{\partial t}+\rho_1\left(\bu_1\cdot\nabla\right)\bu_1
=
\left(a_1\rho_1+a_2\rho_2\right)\bE+\left(a_1\rho_1\bu_1+a_2\rho_2\bu_2\right)\times\bB-\nabla\mathsf{p}_1
\end{equation}
Also, upon assuming neutrality by letting $\epsilon_0\to0$, the electromagnetic fields satisfy the equations
\begin{align}
&\sum_s a_s\rho_s\bu_s=\frac1{\mu_0}\nabla\times\bB-\,a_h\int\!\bp\,f\,\de^3\bp
\\
&
\frac{\partial\bB}{\partial t}=-\nabla\times\bE
\\
&
\sum_s a_s\rho_s=-q_h\int f\,\de^3\bp
\,,\qquad
\operatorname{div}\bB=0
\end{align}
Then, equation \eqref{TotMom} becomes
\begin{equation}
\rho\frac{\partial\bu}{\partial t}+\rho\left(\bu\cdot\nabla\right)\bu
=
-\left(q_h\int\! f\,\de^3\bp\right)\bE+\left(\frac1{\mu_0}\nabla\times\bB-\,a_h\int\!\bp\,f\,\de^3\bp\right)\times\bB-\rho\nabla\mathsf{p}
\,,
\end{equation}
where we have dropped labels for convenience. Finally, inserting Ohm's ideal law $\bE+\bu\times\bB=0$, the kinetic two-fluid system becomes
\begin{align}\label{cc-hybrid-momentum}
&\rho\frac{\partial\bu}{\partial t}+\rho\left(\bu\cdot\nabla\right)\bu
=
\left(q_h\, \bu\int\! f\,\de^3\bp-\,a_h\int\!\bp\,f\,\de^3\bp+\frac1{\mu_0}\nabla\times\bB\right)\times\bB-\rho\nabla\mathsf{p}
\\
& \frac{\partial\rho}{\partial
t}+\operatorname{div}\left(\rho\bu\right)=0 \label{cc-hybrid-mass}
\\
&
\frac{\partial f}{\partial t}+\frac{\bp}{m_h}\cdot\frac{\partial f}{\partial \bq}+q_h\left(\frac{\bp}{m_h}-\bu\right)\times\bB\cdot\frac{\partial f}{\partial \bp}=0
\\
&
\frac{\partial\bB}{\partial t}=\nabla\times\left(\bu\times\bB\right)
\label{cc-hybrid-end}
\,,
\end{align}
which is identical to the current-coupling hybrid scheme presented
in \cite{FuPark,ParkEtAl,BeDeCh}, except the fact that particle
dynamics is governed by the Vlasov equation rather than its
gyrokinetic counterpart. Notice that the above system does not make
any assumption on the form of the Vlasov distribution for the
energetic particles. Therefore this system applies to a whole
variety of possible situations.

We now turn our attention to the energy balance, that is we ask whether the above current-coupling system possesses a Hamiltonian structure. Remarkably, a positive answer is provided by the following Poisson bracket
\begin{align}\nonumber
\{F,G\}=& \int{\bm}\cdot \left[\frac{\delta
F}{\delta{\bm}},\frac{\delta G}{\delta
{\bm}}\right]\de^3\bq
-
\int \rho \left(\frac{\delta F}{\delta \bm}\cdot\nabla\frac{\delta G}{\delta \rho}-\frac{\delta G}{\delta \bm}\cdot \nabla\frac{\delta F}{\delta \rho}\right)\de^3\bq
\\\nonumber
&+
q_h
\int\! f \,\bB\cdot\left(\frac{\delta F}{\delta \bm}\times
\frac{\delta G}{\delta \bm}-\frac{\delta F}{\delta \bm}
\times
\frac{\partial}{\partial \bp}\frac{\delta
G}{\delta f}+\frac{\delta G}{\delta \bm}
\times
\frac{\partial}{\partial \bp}\frac{\delta
F}{\delta f}\right)\dvol
\\\nonumber
&+
\int \!f\left(\left\{\frac{\delta
F}{\delta f},\frac{\delta G}{\delta f}\right\}+q_h\,\bB\cdot\frac{\partial}{\partial \bp}\frac{\delta
F}{\delta f}
\times
\frac{\partial}{\partial \bp}\frac{\delta
G}{\delta f}\right)\dvol
\\
&
-
\int \bB\cdot\left(\frac{\delta F}{\delta \bm}\times\nabla\times\frac{\delta G}{\delta \bB}-\frac{\delta G}{\delta \bm}\times\nabla\times\frac{\delta F}{\delta \bB}\right)\de^3\bq
\,,
\label{PB-current-hybridMHD}
\end{align}
together with the Hamiltonian
\begin{multline}
H(\bm,\rho,f,\bB)=\frac1{2}\int \frac{\left|\bm\right|^2}{\rho}\,\de^3\bq+\frac1{2\m_h}\int f \left|\bp\right|^2\dvol+\int\! \rho\ \mathcal{U}(\rho)\,\de^3\bq
+\frac1{2\mu_0}\int \left|\bB\right|^2\de^3\bq
\,,
\label{Ham-preMHD}
\end{multline}
where the fluid velocity is replaced by the fluid momentum
$\bm=\rho\bu$. The nature of the Hamiltonian structure
\eqref{PB-current-hybridMHD} will be discussed later in relation to
the MHD bracket (see theorem \ref{theorem1}).

\subsection{Pressure-coupling hybrid MHD schemes}
In this section we show how the Hamiltonian structure of the previous current-coupling scheme represents a helpful tool for the formulation of a pressure-coupling scheme. This scheme establishes an equation for the total momentum
\[
\bM=\bm+\int\bp\,f\,\de^3\bp\,,
\]
under the assumption that the averaged kinetic momentum $\bK=\int\!\bp\,f\,\de^3\bp$ does not contribute to the total energy of the system. The definition of the momentum coordinate $\bM$ yields the simple functional relations
\[
\frac{\delta F}{\delta \bm}=\frac{\delta \bar{F}}{\delta \bM}
\,,\qquad\quad\frac{\delta F}{\delta \rho}=\frac{\delta \bar{F}}{\delta \rho}
\,,\qquad\quad\frac{\delta F}{\delta \bB}=\frac{\delta \bar{F}}{\delta \bB}
\,,\qquad\quad
\frac{\delta F}{\delta f}=\frac{\delta \bar{F}}{\delta f}+\bp\cdot\frac{\delta \bar{F}}{\delta \bM}
\]
which can be easily substituted in the bracket \eqref{PB-current-hybridMHD} to produce the new Hamiltonian structure
\begin{align}\nonumber
\{F,G\}=& \int{\bM}\cdot \left[\frac{\delta
F}{\delta{\bM}},\frac{\delta G}{\delta
{\bM}}\right]\de^3\bq
-
\int \rho \left(\frac{\delta F}{\delta \bM}\cdot\nabla\frac{\delta G}{\delta \rho}-\frac{\delta G}{\delta \bM}\cdot \nabla\frac{\delta F}{\delta \rho}\right)\de^3\bq
\\\nonumber
&
+
\int f\left(\left\{\frac{\delta
F}{\delta f},\frac{\delta G}{\delta f}\right\}
+q_h\,
\bB\cdot\frac{\partial}{\partial \bp}\frac{\delta
F}{\delta f}
\times
\frac{\partial}{\partial \bp}\frac{\delta
G}{\delta f}
\right)\dvol
\\\nonumber
&+\int f\left(\left\{\frac{\delta
F}{\delta f},\bp\cdot\frac{\delta G}{\delta \bM}\right\}-\left\{\frac{\delta
G}{\delta f},\bp\cdot\frac{\delta F}{\delta \bM}\right\}\right)\dvol
\\ \label{PB-pressure-hybridMHD}
&
+
\int \bB\cdot\left(\frac{\delta F}{\delta \bM}\times\nabla\times\frac{\delta G}{\delta \bB}-\frac{\delta G}{\delta \bM}\times\nabla\times\frac{\delta F}{\delta \bB}\right)\de^3\bq
\end{align}
Let us now express the Hamiltonian \eqref{Ham-preMHD} in terms of the total momentum $\bM$: we obtain
\begin{equation}\label{exact-hamiltonian}
H=\frac12\int \frac{\left|\bM-\bK\right|^2}{\rho}\, \de^3\bq\ +\frac1{2\m_h}\int f \left|\bp\right|^2\dvol
+ \int\!\rho\,\, \mathcal{U}(\rho)\,\de^3\bq
+\frac1{2\mu_0}\int \left|\bB\right|^2\de^3\bq
\,.
\end{equation}
Then,
the Poisson bracket \eqref{PB-pressure-hybridMHD} yields the velocity equation
\begin{align*}
\frac{\partial \bU}{\partial t}+(\bu\cdot\nabla)\bu=&-\frac1\rho\nabla{\sf p}-\frac1{\m_h \rho}\operatorname{div}\int\!\bp\bp\, f\,\de^3\bp
+\frac1{\mu_0 \rho}\operatorname{curl}\bB\times \bB
\,,
\end{align*}
while the other equations remain identical to those in
\eqref{cc-hybrid-mass}-\eqref{cc-hybrid-end}. Before proceeding
further, we remark that neglecting all terms involving the averaged
kinetic momentum $\int\!\bp\, f\,\de^3\bp$ (so that
$\bu\cdot\nabla\bu\simeq\bU\cdot\nabla\bU$) produces the hybrid MHD
model in \cite{Cheng} (although the general Vlasov equation is
adopted here, rather than a gyrokinetic equation). However, this
crucial step breaks the energy-conserving nature of the system and,
when an energy balance equation is required, the model needs
substantial modifications.

\subsubsection{First pressure-coupling MHD scheme} A structure-preserving approximation can be easily obtained by neglecting the averaged kinetic momentum directly in the expression of the hybrid Hamiltonian \eqref{exact-hamiltonian} (so that $\bM\simeq\bm$), which then becomes
\begin{multline}\label{PB-pressure-hybridMHD-Hamiltonian}
H(\bM,\rho,f,\bB)=\frac12\int \frac{\left|\bM\right|^{2}}\rho\de^3\bq\ +\frac1{2\m_h}\int f \left|\bp\right|^2\dvol
+ \int \!\rho\ \mathcal{U}(\rho)\,\de^3\bq
+\frac1{2\mu_0}\int \left|\bB\right|^2\de^3\bq
\end{multline}
Then, the Poisson bracket \eqref{PB-pressure-hybridMHD}  yields the final equations
\begin{align}\label{hybridMHD1}
&\frac{\partial \bU}{\partial t}+(\bU\cdot\nabla)\bU=-\frac1\rho\nabla{\sf p}-\frac1{\m_h \rho}\operatorname{div}\int \bp\bp f\,\de^3\bp
+\frac1{\mu_0 \rho}\operatorname{curl}\bB\times \bB
\\\label{hybridMHD2}
&\frac{\partial f}{\partial
t}+\left(\boldsymbol{U}+\frac{\bp}{m_h}\right)\cdot\frac{\partial
f}{\partial \bq}-\frac{\partial f}{\partial
\bp}\cdot\nabla\boldsymbol{U}\cdot\bp+a_h\,{\bp}\times \bB\cdot\frac{\partial f}{\partial\bp}=0
\\\label{hybridMHD3}
&\frac{\partial \rho}{\partial t}+\operatorname{div}(\rho\,\bU)=0
\,,\qquad
\frac{\partial \bB}{\partial t}=\operatorname{curl}\left(\bU\times\bB\right)
\,,
\end{align}
where the term $\{f,\bp\cdot\boldsymbol{U}\}$ again appears to
balance the pressure term in the energy conservation. Except for the
Vlasov equation \eqref{hybridMHD2}, the above system is identical to
the scheme presented in \cite{FuPark,ParkEtAl}.

\subsubsection{Second pressure-coupling MHD scheme} Notice that the following expression of the Hamiltonian is also possible
\begin{multline}\label{hybrid-MHD-pc2-Hamiltonian}
H(\bM,\rho,f,\bB)=\frac12\int \frac{\left|\bM\right|^{2}}\rho\de^3\bq\ +\frac1{2\m_h}\int\! f \left|\bp-\frac{\bK}n\right|^2\dvol
+ \int \!\rho\ \mathcal{U}(\rho)\,\de^3\bq
+\frac1{2\mu_0}\int \left|\bB\right|^2\de^3\bq
\,,
\end{multline}
which in turn yields the equations
\begin{align}\label{hybrid-MHD-pc2-velocity}
&\frac{\partial \bU}{\partial t}+(\bU\cdot\nabla)\bU=-\frac1\rho\nabla{\sf p}-\frac1{\m_h \rho}\operatorname{div}\!\int\!\left(\bp-\boldsymbol{V}\right)^{\otimes\,2\!}f\,\de^3\bp
+\frac1{\mu_0 \rho}\operatorname{curl}\bB\times \bB
\\
&\frac{\partial f}{\partial t}+\left(\frac{\bp-\boldsymbol{V}}{m_h}+\boldsymbol{U}\right)\cdot\frac{\partial
f}{\partial \bq}-\frac{\partial f}{\partial
\bp}\cdot\left(\frac{\partial
\boldsymbol{U}}{\partial\bq}\cdot\bp+\frac{\partial
\boldsymbol{V}}{\partial\bq}\cdot\frac{\bp-\boldsymbol{V}}{m_h}\right)+a_h\left({\bp}-\boldsymbol{V}\right)\times \bB\cdot\frac{\partial f}{\partial\bp}=0
\\
&\frac{\partial \rho}{\partial t}+\operatorname{div}(\rho\,\bU)=0
\,,\qquad
\frac{\partial \bB}{\partial t}=\operatorname{curl}\left(\bU\times\bB\right)
\,.
\label{hybrid-MHD-pc2-end}
\end{align}
where $\boldsymbol{V}=\bK/n$. Notice that although complicated, the
above kinetic equation yields the simple advection equation
$\partial_t n+\operatorname{div}(n\,\boldsymbol{U})=0$ for the
zero-th moment $n=\int\!f\,\de^3\bp$. Then, the total density
$D=\rho+m_h\,n$ also satisfies $\partial_t
D+\operatorname{div}(D\,\boldsymbol{U})=0$, which is the total
density equation appearing Cheng's hybrid model \cite{Cheng}. This
transport equation for $D$ and the velocity equation above both
appear exactly the same in the hybrid scheme presented in
\cite{KimEtAl,TaBrKi}. As a conclusion to this section, we remark
that in the present Hamiltonian setting a purely advection equation
for $D$ is necessarily accompanied by the presence of the relative
pressure tensor
$\int\left(\bp-\boldsymbol{V}\right)^{\otimes\,2\!}f\,\de^3\bp$ in
the velocity equation, instead of the absolute pressure
$\int\bp^{\otimes\,2\!}\,f\,\de^3\bp$. This fact is suggestive that
the scheme in \cite{KimEtAl,TaBrKi} is preferable to the one
presented in \cite{Cheng}

\subsection{The Hamiltonian structure and its consequences}\label{PBprinciple}
In the previous sections, the current-coupling hybrid kinetic scheme
\eqref{cc-hybrid-momentum}-\eqref{cc-hybrid-end} was derived by
inserting neutrality and Ohm's ideal law in the general Hamilton's
equations arising from the Poisson bracket
\eqref{PB-kinetic-multifluid}. Then, upon assuming negligible
density and momentum for the energetic component, a simple momentum
shift produced the pressure-coupling hybrid schemes. This section
explains the origin of the hybrid Hamiltonian structures, starting
from a direct sum Poisson bracket. In particular, the reason why the
hybrid Poisson brackets \eqref{PB-current-hybridMHD} and
\eqref{PB-pressure-hybridMHD} both satisfy the Jacobi identity is
not at all evident and it needs particular care.

We shall proceed by first deriving the bracket \eqref{PB-current-hybridMHD} from the direct sum bracket
\begin{align}\nonumber
\{F,G\}=& \int{\bN}\cdot \left[\frac{\delta
F}{\delta{\bN}},\frac{\delta G}{\delta {\bN}}\right]\de^3\bq - \int
\rho \left(\frac{\delta F}{\delta \bN}\cdot\nabla\frac{\delta
G}{\delta \rho}-\frac{\delta G}{\delta \bN}\cdot \nabla\frac{\delta
F}{\delta \rho}\right)\de^3\bq
\\ \nonumber
& - \int \mathbf{A}\cdot \left(\frac{\delta F}{\delta
\bm}\operatorname{div}\frac{\delta G}{\delta
\mathbf{A}}+\frac{\delta G}{\delta
\bN}\operatorname{div}\frac{\delta F}{\delta
\mathbf{A}}-\nabla\times\left(\frac{\delta F}{\delta
\bN}\times\frac{\delta G}{\delta \mathbf{A}}-\frac{\delta G}{\delta
\bN}\times\frac{\delta F}{\delta \mathbf{A}}\right)\right)\de^3\bq
\\
&
+
\int \hat{f}\left\{\frac{\delta
F}{\delta \hat{f}},\frac{\delta G}{\delta \hat{f}}\right\}\dvol
\label{PB-MHD+Vlasov}
\,,
\end{align}
where the first two lines coincide with the bracket of ideal
compressible MHD \cite{HolmKupershmidt,MorrisonGreene,Morrison2},
while the last term determines the Poisson bracket of the Vlasov
equation. The momentum variable $\bN$ is a new variable that will be
given physical meaning in theorem \ref{theorem1}. Notice that in the
above bracket, the Vlasov distribution $\hat{f}(\bq,\boldsymbol{p})$
is expressed as a function of the mechanical momentum variable
${\boldsymbol{p}=m_h\mathbf{v}+q_a\mathbf{A}}$ and \emph{not} as a
function of the variable $\bp=m_h\mathbf{v}$. Thus one has
\[
\hat{f}(\bq,\boldsymbol{p})=\hat{f}(\bq,\bp+q_h\mathbf{A})=f(\bq,\bp)
\]
(Notice the different notation: the mechanical momentum is denoted
by the italic bold font $\boldsymbol{p}$, while the roman bold
notation $\bp$ corresponds to $m_h \mathbf{v}$.)

\begin{remark}[Lie algebra of the current-coupling scheme]
The Poisson bracket \eqref{PB-MHD+Vlasov} is a Lie-Poisson bracket
on the direct-sum Lie algebra
\[
C^\infty(\Bbb{R}^6)\oplus\Big(\mathfrak{X}(\Bbb{R}^3)\,\circledS\left(C^\infty(\Bbb{R}^3)\oplus\Omega^2(\Bbb{R}^3)\right)\Big)
\]
where $C^\infty(\Bbb{R}^6)$ is the algebra of Hamiltonian functions
(dual to Vlasov distributions), while the variables
$(\bN,\rho,\mathbf{A})$ belong to the dual space of the
semidirect-product Lie algebra
\[
\mathfrak{X}(\Bbb{R}^3)\,\circledS{\left(C^\infty(\Bbb{R}^3)\oplus\Omega^2(\Bbb{R}^3)\right)}
\,.
\]
Here $\Omega^2(\Bbb{R}^3)$ is the space of differential two-forms
containing the electric induction $\bD\cdot\de\mathbf{S}$ (where
$\de\mathbf{S}=\hat{\mathbf{n}}\,\de S$ denotes the product of an
infinitesimal surface element $\de S$ with its normal unit vector
$\hat{\mathbf{n}}$).
\end{remark}

The Poisson bracket \eqref{PB-current-hybridMHD} is equivalent to \eqref{PB-MHD+Vlasov} in the following sense:
\begin{theorem}\label{theorem1}
The Poisson bracket \eqref{PB-MHD+Vlasov} transforms to the bracket \eqref{PB-current-hybridMHD} under the momentum shift
\[
\bm=\bN+q_h\mathbf{A}\int\!\hat{f}(\bq,\boldsymbol{p})\,\de^3\boldsymbol{p}=\bN+q_h\mathbf{A}\int{f}(\bq,\bp)\,\de^3\bp
\,.
\]
\end{theorem}
This result is proved in the appendix \ref{appendix1}. Perhaps the most direct physical consequence resides in the following circulation theorem
\begin{corollary}[Kelvin's theorem for the current-coupling scheme]\label{corollary1} For every loop $\gamma_t(\bu)$ moving with velocity $\bu=\bm/\rho$,
the current-coupling scheme \eqref{cc-hybrid-momentum}-\eqref{cc-hybrid-end} possesses the following circulation
theorem
\begin{multline*}
\frac{\de}{\de
t}\oint_{\gamma_t(\bu)}\left(\bu-q_h\,\frac{n}{\rho}\,\mathbf{A}\right)\cdot\de\bq=-\,\oint_{\gamma_t(\bu)}\frac1\rho\,
\bB\times\!\Big(\mu_0^{-1}\nabla\times\bB-a_h\left(\bK-m_h\,n\bu\right)\!\Big)\cdot\de\bq
\\+a_h\oint_{\gamma_t(\bu)}\frac1\rho\,\big(\nabla\cdot\left(\bK-m_h\,n\bu\right)\big)\,\mathbf{A}\cdot\de\bq\,,
\end{multline*}
where $\bK=\int\!\bp\,f\,\de^3\bp$.
\end{corollary}
\begin{proof}
From theorem \ref{theorem1}, one concludes that the current-coupling scheme
\eqref{cc-hybrid-momentum}-\eqref{cc-hybrid-end} possesses the
Poisson bracket \eqref{PB-MHD+Vlasov} and the Hamiltonian \eqref{Ham-preMHD} in the form
\begin{multline}\label{HamiltonianN}
H=\frac1{2}\int
\frac1\rho\left|\bN+q_h\,\mathbf{A}\!\int\!\hat{f}\,\de^3\boldsymbol{p}\right|^2\de^3\bq+\frac1{2 m_h}\int\!
\hat{f}\,
\left|\boldsymbol{p}-q_h\,\mathbf{A}\right|^2\de^3\bq\,\de^3\boldsymbol{p}
\\
+\int\! \rho\ \mathcal{U}(\rho)\,\de^3\bq
+\frac1{2\mu_0}\int \left|\nabla\times\mathbf{A}\right|^2\de^3\bq
\,.
\end{multline}
Then, replacing the Hamiltonian \eqref{HamiltonianN} in the bracket
\eqref{PB-MHD+Vlasov} yields the momentum equation
\begin{equation}
\left(\frac{\partial}{\partial t}+\pounds_\bu\right)\bN=-
\rho\,\nabla\frac{\delta H}{\delta \rho} +\frac{\delta H}{\delta
\bA}\times\left(\nabla\times\bA\right)-\bA\,\nabla\cdot\frac{\delta
H}{\delta \bA} \label{Nequation}
\end{equation}
where $\pounds_\bu$ denotes the Lie derivative along the velocity
vector field $\bu$. After dividing by $\rho$ and calculating
\begin{align*}
\frac{\delta H}{\delta
\bA}=&q_h\,\frac{n}{\rho}\left(\bN+q_h\,n\,\bA\right)+\mu_0^{-1}\nabla\times\nabla\times\bA-a_h\int\!\hat{f}\left(\boldsymbol{p}-q_h\,\bA\right)\de^3\boldsymbol{p}
\\
=&\mu_0^{-1}\nabla\times\bB-a_h\left(\bK-m_h\,n\bu\right)\,,
\end{align*}
the proof follows easily by taking the circulation integral on both
sides of \eqref{Nequation}.
\end{proof}\\
Another relevant consequence of theorem \ref{theorem1} is the existence of Casimir functionals
\begin{corollary}[Casimir functionals]
The Poisson bracket \eqref{PB-current-hybridMHD} possesses the following Casimir functional
\[
C(f)=\int\!\Phi(f)\,\dvol\,.
\]
\end{corollary}
\begin{proof}
Although it can be proven by direct verification, this is a direct consequence of the Casimir functional $C(\hat{f})=\int\Lambda(\hat{f})\,\de^3\bq\,\de^3\boldsymbol{p}$ for the Poisson bracket \eqref{PB-MHD+Vlasov}, which is inherited from the general properties of the Hamiltonian structure of the Vlasov equation.
\end{proof}

\smallskip
\noindent
Notice that the magnetic helicity $\mathcal{H}=\int\mathbf{A}\cdot\nabla\times\mathbf{A}\ \de^3\bq$ is also a Casimir for the bracket \eqref{PB-MHD+Vlasov}, although it is not strictly a Casimir for the bracket \eqref{PB-current-hybridMHD}, which is rather expressed in terms of the magnetic induction $\bB=\nabla\times\mathbf{A}$.

\medskip

Analogous properties as those above also hold for the pressure coupling scheme. This can be seen by the following relation between the Poisson brackets \eqref{PB-pressure-hybridMHD} and \eqref{PB-MHD+Vlasov}:
 \begin{theorem}\label{theorem2}
The Poisson bracket \eqref{PB-MHD+Vlasov} transforms to the bracket \eqref{PB-pressure-hybridMHD} under the change of variable
\begin{align*}
\bM=&\bN+\int\!\boldsymbol{p}\,\hat{f}(\bq,\boldsymbol{p})\,\de^3\boldsymbol{p}=\bN+\int\!\bp\,{f}(\bq,\bp)\,\de^3\bp+q_h\mathbf{A}\int{f}(\bq,\bp)\,\de^3\bp
\,.
\end{align*}
\end{theorem}
The proof of the above theorem is presented in appendix \ref{appendix2}. Again, the above statement yields the 
\begin{corollary}[Kelvin's theorems for pressure-coupling schemes]
The Kelvin circulation laws for the pressure coupling schemes \eqref{hybridMHD1}-\eqref{hybridMHD3} and \eqref{hybrid-MHD-pc2-velocity}-\eqref{hybrid-MHD-pc2-end} read respectively as
\begin{multline*}
\frac{\de}{\de
t}\oint_{\gamma_t(\bU)}\left(\bU-\frac1\rho\int\!\boldsymbol{p}\,\hat{f}(\bq,\boldsymbol{p})\,\de^3\boldsymbol{p}\right)\cdot\de\bq=
-\,\oint_{\gamma_t(\bU)}\frac1\rho\,
\bB\times\Big(\mu_0^{-1}\nabla\times\bB-a_h\bK\Big)\cdot\de\bq
\\
+a_h\oint_{\gamma_t(\bU)}\frac1\rho\,(\nabla\cdot\bK)\,\mathbf{A}\cdot\de\bq\,,
\end{multline*}
\begin{equation*}
\frac{\de}{\de
t}\oint_{\gamma_t(\bU)}\left(\bU-\frac1\rho\int\!\boldsymbol{p}\,\hat{f}(\bq,\boldsymbol{p})\,\de^3\boldsymbol{p}\right)\cdot\de\bq=
-\,{\mu_0^{-1}}\oint_{\gamma_t(\bU)}\, \rho^{-1}\,
\bB\times\nabla\times\bB\cdot\de\bq
\,,\hspace{2.4cm}
\end{equation*}
where $\bK=\int \bp\,f\,\de^3\bp$ and the loop $\gamma_t(\bU)$ moves with the velocity
$\bU=\bM/\rho$.
\end{corollary}
\begin{proof}
Because of theorem \ref{theorem2}, the pressure coupling scheme \eqref{hybridMHD1}-\eqref{hybridMHD3} can be written in terms of the momentum variable $\bN=\bM-\int\!\boldsymbol{p}\,\hat{f}(\bq,\boldsymbol{p})\,\de^3\boldsymbol{p}$. Indeed, the equation for $\bN$ is given immediately by the Poisson bracket \eqref{PB-MHD+Vlasov} together with the Hamiltonian \eqref{PB-pressure-hybridMHD-Hamiltonian} in the form
\begin{multline}\label{HamiltonianPCS-N}
H=\frac1{2}\int
\frac1\rho\left|\bN+\int\!\boldsymbol{p}\,\hat{f}(\bq,\boldsymbol{p})\,\de^3\boldsymbol{p}\right|^2\de^3\bq+\frac1{2\m_h}\int\!
\hat{f}\,
\left|\boldsymbol{p}-q_h\,\mathbf{A}\right|^2\de^3\bq\,\de^3\boldsymbol{p}
\\
+\int\! \rho\ \mathcal{U}(\rho)\,\de^3\bq
+\frac1{2\mu_0}\int \left|\nabla\times\mathbf{A}\right|^2\de^3\bq
\,.
\end{multline}
At this point, the Poisson bracket \eqref{PB-MHD+Vlasov} produces the momentum equation \eqref{Nequation}, while the Hamiltonian \eqref{HamiltonianPCS-N} yields
\begin{equation*}
\frac{\delta H}{\delta
\bA}= {\mu_0}^{-1}\,\nabla\times\bB-a_h\,\bK\,.
\end{equation*}
Then, after replacing the above variational derivative in equation \eqref{Nequation}, the remaining steps coincide with those already followed in the proof of corollary \ref{corollary1}.

\noindent
The same arguments also hold for the second pressure-coupling scheme  \eqref{hybrid-MHD-pc2-velocity}-\eqref{hybrid-MHD-pc2-end}. In this case, the Hamiltonian \eqref{hybrid-MHD-pc2-Hamiltonian} is written explicitly in terms of $\bN$ as
\begin{multline}\label{HamiltonianPCS2-N}
H=\frac1{2}\int
\frac1\rho\left|\bN+\int\!\boldsymbol{p}\,\hat{f}(\bq,\boldsymbol{p})\,\de^3\boldsymbol{p}\right|^2\de^3\bq+\frac1{2\m_h}\int\!
\hat{f}\,
\left|\boldsymbol{p}-q_h\,\mathbf{A}\right|^2\de^3\bq\,\de^3\boldsymbol{p}
\\
-\frac1{2m_h}\int\frac{1}{\int\!\hat{f}(\bq,\boldsymbol{p})\,\de^3\boldsymbol{p}}\,\left|\int\!\boldsymbol{p}\,\hat{f}(\bq,\boldsymbol{p})\,\de^3\boldsymbol{p}-q_h\bA\int\!\hat{f}(\bq,\boldsymbol{p})\,\de^3\boldsymbol{p}\right|^2\de^3\bq
\\
+\int\! \rho\ \mathcal{U}(\rho)\,\de^3\bq
+\frac1{2\mu_0}\int \left|\nabla\times\mathbf{A}\right|^2\de^3\bq
\,,
\end{multline}
which in turn yields the following variational derivative:
\begin{equation*}
\frac{\delta H}{\delta
\bA}= {\mu_0}^{-1}\,\nabla\times\bB\,.
\end{equation*}
Then, inserting the following relation in the equation \eqref{Nequation} arising from the Poisson bracket \eqref{PB-MHD+Vlasov} and following the same steps as in the previous cases completes the derivation of the circulation law for the second pressure-coupling scheme  \eqref{hybrid-MHD-pc2-velocity}-\eqref{hybrid-MHD-pc2-end}.
\end{proof}

\noindent
Notice that if the hot particle density $n$ is negligible at all times, then the above circulation law for the first pressure-coupling scheme \eqref{hybridMHD1}-\eqref{hybridMHD3} reduces to the circulation law given in corollary \ref{corollary1} for the current coupling scheme. However, this particular case presents some obstacles that will be discussed in the next section.

It is easy to
verify that functionals of the type $\int\Phi(f)\,\dvol$ are also
Casimirs of the Poisson bracket \eqref{PB-pressure-hybridMHD}. The
identification of Casimir functionals is a fundamental feature of
the Hamiltonian approach, which allows for a nonlinear stability
analysis as extensively described in \cite{HoMaRaWe}. The study of
the stability properties of the hybrid MHD schemes is a considerable
part of future work plans.

\smallskip
\noindent Thus, we have the following consequence
\begin{corollary}
The brackets \eqref{PB-current-hybridMHD} and \eqref{PB-pressure-hybridMHD} are Poisson brackets, whose Jacobi identity is inherited from the direct sum Poisson bracket \eqref{PB-MHD+Vlasov}.
\end{corollary}

\begin{remark}[Lie algebra of pressure-coupling schemes]
From an algebraic point of view, the Poisson bracket
\eqref{PB-pressure-hybridMHD}  is a Lie Poisson bracket
\cite{MaWeRa} on the semidirect-product Lie algebra
\[
\mathfrak{X}(\mathbb{R}^3)\,\circledS\left(C^\infty(\mathbb{R}^3) \times\Omega^1(\mathbb{R}^3)\times C^\infty(\mathbb{R}^6)\right)
\]
where we have denoted by $\mathfrak{X}(\mathbb{R}^3)$ the space of
velocity vector fields, while $\Omega^1(\mathbb{R}^3)$ denotes the
space of differential one-forms (i.e. the space of magnetic vector
potentials). Notice that the space $C^\infty(\mathbb{R}^6)$ is a
Poisson algebra, which is dual to the  Poisson manifold of Vlasov
distributions. This construction containing the semidirect-product
Lie algebra
$\mathfrak{X}(\mathbb{R}^3)\,\circledS\,C^\infty(\mathbb{R}^6)$
appears in a continuum model for the first time.
\end{remark}

\subsection{Arguments in favor of the current-coupling scheme}
The current coupling scheme for hybrid MHD is often considered to be
 equivalent to ordinary pressure coupling schemes
\cite{ParkEtAl,PaBeFuTaStSu,FuPark}. However, the assumption of a
rarefied energetic ion plasma, poses nontrivial problems in the
mathematical formulation. Indeed, we recognize that the assumption
of negligible particle density $n$ may lead to some ambiguities. For
example, we saw that in the second pressure-coupling scheme
\eqref{hybrid-MHD-pc2-velocity}-\eqref{hybrid-MHD-pc2-end}, the
condition $n=0$ is preserved by the dynamics, although this reflects
in a possible divergence of the Hamiltonian, which indeed contains
the term $n^{-1}\left|\bK\right|^2$. Then, the only possibility to avoid this
divergence is to set $\bK=0$. For example, the whole velocity
$\bK/n$ is sometimes neglected completely \cite{TaBrKi}. However,
the well known pressure term
$\operatorname{div}\!\int\left(\bp-\boldsymbol{V}\right)^{\otimes\,2\!}f\,\de^3\bp$
appearing in the equation for $\bK$ forbids this possibility unless
one accepts that this pressure term is also negligible, which
contradicts the starting hypothesis of non-negligible pressure
effects in the total momentum dynamics. This non-negligible pressure
term is a continuous source of hot particle averaged momentum, whose
initial low levels are not maintained by the dynamics. Notice that
this phenomenon is not peculiar of the pressure-coupling schemes
presented in this paper; rather, this problem is intrinsically
unavoidable and it is  present in all pressure-coupling schemes.

The current-coupling scheme provides an exceptional way of avoiding
this sort of problems, since it makes no assumption on the moments
of the Vlasov distribution. Some advantages of the current-coupling
scheme were also emphasized in \cite{BeDeCh}. Moreover, the
existence of the newly discovered Hamiltonian structure for this
model, together with its consequent circulation theorem and Casimir functions, enriches
this model significantly. In conclusion, the
pressure coupling scheme is very different from the current coupling
scheme and the latter solves all the significant problems emerging from
pressure coupling schemes.

\section{Hybrid kinetic-fluid formulation of Hall MHD}
In the preceding sections we have investigated the hybrid
kinetic-fluid formulation of MHD models taking into account the
dynamics of energetic ions. On the other hand, the existence of
energetic ions is also a peculiar feature of Hall effects in
magnetized plasmas. Indeed, one of the main consequences of Hall
effects is the progressive acceleration of ions that become
particularly fast. In Hall MHD models, this scenario leads to the
simple idea of treating the whole ion dynamics by a kinetic
approach, while the electrons are still considered as a neutralizing
background. This model is usually referred to as
``kinetic-particle-ion/fluid electron'' hybrid scheme
\cite{PaBeFuTaStSu,WiYiOmKaQu,ByCoCoHa,YiWiEtAl}. Its equations are easily
derived from the kinetic-multifluid system of section
\ref{kinetic-mf-system}, upon letting the only label $s=2$ identify
the electron species. Then, the plasma equations of motion reduce to
the Vlasov kinetic equation \eqref{multi-fluid-Vlasov} and the
electron fluid equation \eqref{multi-fluid-momentum} (with $s=2$).
Neglecting electron inertia (that is  letting $m_e\to0$) and
assuming neutrality (that is letting $\epsilon_0\to0$) yields the
well known equations \cite{PaBeFuTaStSu,WiYiOmKaQu,ByCoCoHa}
\begin{align}\label{kinetic-HMHD1}
&\frac{\partial f}{\partial t}+\frac{\bp}{m_i}\cdot\frac{\partial f}{\partial \bq}+a_i{\bp}\times\bB\cdot\frac{\partial f}{\partial \bp}
+\frac{q_i}{a_e\rho_e}\left[a_e\nabla\mathsf{p}_e-\left(a_i\!\int\!\bp\,f\,\de^3\bp-\mu_0^{-1}\nabla\times\bB\right)\times\bB\right]\cdot\frac{\partial f}{\partial \bp}=0
\\\label{kinetic-HMHD2}
&
\frac{\partial \bB}{\partial t}=\nabla\times\left[\frac1{a_e\rho_e}\,\bB\times\!\left(a_i\!\int\!\bp\,f\,\de^3\bp-\mu_0^{-1}\nabla\times\bB\right)\right]
\,,
\end{align}
In the search for a Hamiltonian structure of the above equations, it
is helpful to recall the Hamiltonian structure of Hall MHD, which
was presented in \cite{Holm3}.  Proceeding by analogy to the
Hall-MHD bracket in \cite{Holm3}, one introduces the
Poisson bracket
\begin{align}\nonumber
\{F,G\}=&\int\left(\mathsf{Q}_e^{-1}\,\nabla\times\mathbf{A}\right)\cdot\frac{\delta F}{\delta\mathbf{A}}\times\frac{\delta G}{\delta\mathbf{A}}\ \de^3\bq
-\int\left(\frac{\delta F}{\delta\mathbf{A}}\cdot\nabla\frac{\delta G}{\delta\mathsf{Q}_e}-\frac{\delta G}{\delta\mathbf{A}}\cdot\nabla\frac{\delta F}{\delta\mathsf{Q}_e}\right)\de^3\bq
\\\label{kinetic-HMHD-PB1}
&
+\int \hat{f}\left\{\frac{\delta
F}{\delta \hat{f}},\frac{\delta G}{\delta \hat{f}}\right\}\de^3\bq\,\de^3\boldsymbol{p}
\,,
\end{align}
where $\mathsf{Q}_e=a_e\rho_e$ is the electron charge density and
the kinetic Vlasov term in the last line replaces the fluid bracket
terms appearing in the Hamiltonian structure of Hall MHD
\cite{Holm3}. Notice that the first line of \eqref{kinetic-HMHD-PB1} is equivalent to a fluid bracket in the
density-momentum variables
$\left(\mathsf{Q}_e,\mathsf{Q}_{e\,}\mathbf{A}\right)$ \cite{Holm3}. At this
point, the Hamiltonian structure of the equations
\eqref{kinetic-HMHD1}-\eqref{kinetic-HMHD2} arises by simply
expressing the above bracket in terms of the Vlasov distribution
$f(\bq,\bp)=f(\bq,\boldsymbol{p}-q_i\mathbf{A})=\hat{f}(\bq,\boldsymbol{p})$.
Then, the final result is
\begin{align}\nonumber
\{F,G\}=&\int\left(\mathsf{Q}_e^{-1}\,\nabla\times\mathbf{A}\right)\cdot\frac{\delta F}{\delta\mathbf{A}}\times\frac{\delta G}{\delta\mathbf{A}}\,\de^3\bq
-\int\left(\frac{\delta F}{\delta\mathbf{A}}\cdot\nabla\frac{\delta G}{\delta\mathsf{Q}_e}-\frac{\delta G}{\delta\mathbf{A}}\cdot\nabla\frac{\delta F}{\delta\mathsf{Q}_e}\right)\de^3\bq
\\\nonumber
&
+
\int f\left(\left\{\frac{\delta
F}{\delta f},\frac{\delta G}{\delta f}\right\}
+q_i\,
\nabla\times\mathbf{A}\cdot\frac{\partial}{\partial \bp}\frac{\delta
F}{\delta f}
\times
\frac{\partial}{\partial \bp}\frac{\delta
G}{\delta f}
\right)\dvol
\\\nonumber
&-
q_i\int\! f\left(\left\{\frac{\delta F}{\delta f},\frac{\delta G}{\delta \mathsf{Q}_e}\right\}-\left\{\frac{\delta G}{\delta f},\frac{\delta F}{\delta \mathsf{Q}_e}\right\}\right)\dvol
\\\nonumber
&+\int\!\mathsf{Q}_e^{-1}f\,\nabla\times\mathbf{A}\cdot\!\left(\frac{\partial}{\partial\bp}\frac{\delta G}{\delta f}\times\frac{\delta F}{\delta \mathbf{A}}-\frac{\partial}{\partial\bp}\frac{\delta F}{\delta f}\times\frac{\delta G}{\delta \mathbf{A}}\right)\dvol
\\
&+
q_i^2\int\!\mathsf{Q}_e^{-1}\,\nabla\times\mathbf{A}\cdot\!\left(\int\!f\,\frac{\partial}{\partial \bp}\frac{\delta F}{\delta f}\,\de^3\bp\right)\times\left(\int\!f\,\frac{\partial}{\partial \bp'}\frac{\delta G}{\delta f}\,\de^3\bp'\right)\de^3\bq
\,,
\label{kinetic-HMHD-PB2}
\end{align}
which is indeed the Poisson bracket of equations \eqref{kinetic-HMHD1}-\eqref{kinetic-HMHD2} and it is accompanied by the Hamiltonian
\begin{equation}
H(f,\mathbf{A},\mathsf{Q}_e)=\frac1{2\m_i}\int\! f \left|\bp\right|^2\dvol+\int\! \mathcal{\phi}(\mathsf{Q}_e)\,\de^3\bq
+\frac1{2\mu_0}\int \left|\nabla\times\mathbf{A}\right|^2\de^3\bq
\,,
\end{equation}
where $\phi(\mathsf{Q}_e)$ appropriately denotes the internal energy expressed in terms of the charge $\mathsf{Q}_e$.
Notice that the charge neutrality $\mathsf{Q}_e+q_i\,n=0$ is preserved by the dynamics, as it is shown by the equation
\[
\frac{\partial \mathsf{Q}_e}{\partial t}=a_i\operatorname{div}\!\int \!\bp\,f\,\de^3\bp=-q_i\,\frac{\partial}{\partial t}\!\int \!\bp\,f\,\de^3\bp
\,,
\]
which arises from the Hamiltonian structure \eqref{kinetic-HMHD-PB2}.

It is also interesting to notice that the class of Casimir functionals $C(\hat{f})=\int\!\Lambda(\hat{f})\,\de^3\bq\,\de^3\boldsymbol{p}$ for the Vlasov equation again produces the Casimir functionals $C(f)=\int\Phi(f)\,\dvol$ for the Hamiltonian structure \eqref{kinetic-HMHD-PB1} under consideration. Here, the magnetic helicity
\[
\mathcal{H}={\int\mathbf{A}\cdot\nabla\times\mathbf{A}\,\de^3\bq}
\]
is a Casimir of both Poisson brackets \eqref{kinetic-HMHD-PB1} and \eqref{kinetic-HMHD-PB2}. These Casimir functionals provide the opportunity for a nonlinear stability analysis, which will the the subject of future work.

\begin{remark}[Lie algebraic structure] The Hamiltonian structure \eqref{kinetic-HMHD-PB2} of the hybrid Hall-MHD model is inherited from the Lie-Poisson bracket \eqref{kinetic-HMHD-PB1} on the direct-sum Lie algebra
\[
\left(\mathfrak{X}(\Bbb{R}^3)\,\circledS \,C^\infty(\Bbb{R}^3)\right) \oplus C^\infty(\Bbb{R}^6)
\,,
\]
where $C^\infty(\Bbb{R}^6)$ is the Poisson algebra of Hamiltonian
functions, while the variables
$\left(\mathsf{Q}_e\mathbf{A},\mathsf{Q}_e\right)$ belong to the
dual space of the semidirect-product Lie algebra
$\mathfrak{X}(\Bbb{R}^3)\,\circledS \,C^\infty(\Bbb{R}^3)$ \cite{Holm3}. The Hamiltonian corresponding to the Poisson bracket
\eqref{kinetic-HMHD-PB1} is given by
\begin{equation}
H(\hat{f},\mathbf{A},\mathsf{Q}_e)=\frac1{2\m_i}\int\! \hat{f}\, \left|\boldsymbol{p}-q_i\mathbf{A}\right|^2\de^3\bq\,\de^3\boldsymbol{p}+\int\! \mathcal{\phi}(\mathsf{Q}_e)\,\de^3\bq
+\frac1{2\mu_0}\int \left|\nabla\times\mathbf{A}\right|^2\de^3\bq
\,.
\end{equation}
\end{remark}

The absence of a fluid velocity equation in the hybrid formulation
of Hall MHD makes this model particularly difficult to implement
numerically. However, one can introduce a velocity equation by
simply coupling the ion kinetic equation to the equation for its
first two moments
$(\rho_i\boldsymbol{V}_i,\rho_i)=(\int\!\bp\,f\,\de^3\bp,\,\mathsf{m}_i
\int\!f\,\de^3\bp)$. Evidently, this operation does not produce any
changes to the physics, as long as the form of the Hamiltonian is
not modified. On the other hand, the ion momentum equation might be
useful to isolate kinetic pressure terms explicitly
\cite{PaBeFuTaStSu}. When these kinetic pressure terms are omitted,
the ion momentum equation consistently returns the ordinary
equations of Hall MHD.

\section{Conclusions and future work plans}
This paper has presented the Hamiltonian structure of several hybrid
kinetic-fluid models. Upon making assumptions on the form of the
Hamiltonian, new pressure-coupling schemes were derived for either
ordinary fluids and MHD, whereas the existing current-coupling MHD
scheme \cite{ParkEtAl,PaBeFuTaStSu} was shown to possess an
intrinsic Hamiltonian structure. All Poisson brackets were derived
from first principles, so that the Jacobi identity is always
guaranteed. In turn, the Hamiltonian structure produces a Casimir
functional and a circulation law for each of these models. The same
approach was used to present the Hamiltonian structure of the
existing hybrid Hall-MHD scheme \cite{WiYiOmKaQu,YiWiEtAl}.

The next steps will focus on the nonlinear stability analysis, which
is now made possible by the newly discovered Hamiltonian structures.
Another relevant open question concerns the extension of these
Hamiltonian structures  to the case of gyrokinetic equations, which
are frequently used in simulations. The Hamiltonian structure of
gyrokinetic equations is well known \cite{BrHa,Littlejohn,GrKaLi}
and it is possible that the same methods used in this paper can be
fruitful also in that context.

While the Hamiltonian approach was used in this paper, its
Lagrangian counterpart may also deserve further investigation. This
direction can be easily approached by the methods developed in
\cite{CeMaHo} and is part of current work \cite{HoTr2010}.

As a conclusive remark, it is worth mentioning that while pressure coupling schemes were shown to require transport-like terms in the accompanying kinetic equation, one cannot exclude a priori that
removing these terms still allows for other exotic Poisson brackets. Indeed,  although this is unlikely, it cannot be excluded. However, even if this is the case, the
alternative Poisson brackets would not be of Lie-Poisson form and
thus they would not possess well defined symmetry properties, which are the main focus of this paper.

\subsubsection*{Acknowledgments}
It is a pleasure to thank Darryl Holm and Giovanni Lapenta for
helpful discussions and correspondence. Their hospitality at
Imperial College London and Katholieke Universiteit Leuven is
greatly acknowledged. I am also grateful to Phil Morrison and the
anonymous referees for their valuable advice. In addition, I would
like to thank Emanuele Tassi for many stimulating discussions and for
his hospitality at the Center of Theoretical Physics in Marseille.

\newpage

\appendix

\section{Appendix}

\subsection{Proof of theorem \ref{theorem1}\label{appendix1}} First,
one computes the variational derivatives
\[
\frac{\delta \bar{F}}{\delta \bm}=\frac{\delta{F}}{\delta \bN}
\,,\qquad \frac{\delta{F}}{\delta
\mathbf{A}}=\frac{\delta\bar{F}}{\delta
\mathbf{A}}+q_h\,\frac{\delta \bar{F}}{\delta
\bm}\int\!\hat{f}(\bq,\boldsymbol{p})\,\de^3\boldsymbol{p} \,,\qquad
\frac{\delta{F}}{\delta{f}}=\frac{\delta\bar{F}}{\delta{f}}+q_h\,\mathbf{A}\cdot\frac{\delta
\bar{F}}{\delta \bm}
\]
Then, the Poisson bracket \eqref{PB-MHD+Vlasov} becomes
\begin{align}\nonumber
\{F,G\}=& \int{\bN}\cdot \left[\frac{\delta
F}{\delta{\bm}},\frac{\delta G}{\delta {\bm}}\right]\de^3\bq - \int
\rho \left(\frac{\delta F}{\delta \bm}\cdot\nabla\frac{\delta
G}{\delta \rho}-\frac{\delta G}{\delta \bm}\cdot \nabla\frac{\delta
F}{\delta \rho}\right)\de^3\bq
\\ \nonumber
& - \int \mathbf{A}\cdot \left(\frac{\delta F}{\delta
\bm}\operatorname{div}\frac{\delta G}{\delta
\mathbf{A}}-\frac{\delta G}{\delta
\bm}\operatorname{div}\frac{\delta F}{\delta
\mathbf{A}}-\nabla\times\left(\frac{\delta F}{\delta
\bm}\times\frac{\delta G}{\delta \mathbf{A}}-\frac{\delta G}{\delta
\bm}\times\frac{\delta F}{\delta \mathbf{A}}\right)\right)\de^3\bq
\\ \nonumber
&-q_h \int \mathbf{A}\cdot \left(\frac{\delta F}{\delta
\bm}\operatorname{div}\!\left(n\frac{\delta G}{\delta
\mathbf{m}}\right)-\frac{\delta G}{\delta
\bm}\operatorname{div}\!\left(n\frac{\delta F}{\delta
\mathbf{m}}\right)-2\nabla\times\left(n\,\frac{\delta F}{\delta
\bm}\times\frac{\delta G}{\delta \bm}\right)\right)\de^3\bq
\\\nonumber
& + \int \hat{f}\left\{\frac{\delta F}{\delta \hat{f}},\frac{\delta
G}{\delta \hat{f}}\right\}\dvol +q_h \int
\hat{f}\left(\left\{\frac{\delta F}{\delta
\hat{f}},\mathbf{A}\cdot\frac{\delta G}{\delta
\bm}\right\}-\left\{\frac{\delta G}{\delta \hat{f}},
\mathbf{A}\cdot\frac{\delta F}{\delta \bm}\right\}\right)\dvol
\end{align}
where we have dropped the bar notation for convenience and we have denoted $n=\int\!\hat{f}\,\de^3\boldsymbol{p}$. Upon expanding the third line above, one has
\begin{align}\nonumber
\{F,G\}=& \int{\bN}\cdot \left[\frac{\delta
F}{\delta{\bm}},\frac{\delta G}{\delta {\bm}}\right]\de^3\bq - \int
\rho \left(\frac{\delta F}{\delta \bm}\cdot\nabla\frac{\delta
G}{\delta \rho}-\frac{\delta G}{\delta \bm}\cdot \nabla\frac{\delta
F}{\delta \rho}\right)\de^3\bq
\\ \nonumber
& - \int \mathbf{A}\cdot \left(\frac{\delta F}{\delta
\bm}\operatorname{div}\frac{\delta G}{\delta
\mathbf{A}}-\frac{\delta G}{\delta
\bm}\operatorname{div}\frac{\delta F}{\delta
\mathbf{A}}-\nabla\times\left(\frac{\delta F}{\delta
\bm}\times\frac{\delta G}{\delta \mathbf{A}}-\frac{\delta G}{\delta
\bm}\times\frac{\delta F}{\delta \mathbf{A}}\right)\right)\de^3\bq
\\ \nonumber
&-q_h \int \mathbf{A}\cdot \left(\frac{\delta F}{\delta
\bm}\left(\frac{\delta G}{\delta \bm}\cdot\nabla
n\right)-\frac{\delta G}{\delta \bm}\left(\frac{\delta F}{\delta
\bm}\cdot\nabla n\right)\right)
\\\nonumber
&-q_h \int n\,\mathbf{A}\cdot\left(\frac{\delta F}{\delta
\bm}\operatorname{div}\!\left(\frac{\delta G}{\delta
\mathbf{m}}\right)-\frac{\delta G}{\delta
\bm}\operatorname{div}\!\left(\frac{\delta F}{\delta
\mathbf{m}}\right)\right)\de^3\bq
\\ \nonumber
&+2q_h \int \mathbf{A}\cdot \nabla\times\left(n\,\frac{\delta
F}{\delta \bm}\times\frac{\delta G}{\delta \bm}\right)\de^3\bq +
\int \hat{f}\left\{\frac{\delta F}{\delta \hat{f}},\frac{\delta
G}{\delta \hat{f}}\right\}\dvol
\\ \nonumber
&+q_h \int \hat{f}\left(\left\{\frac{\delta F}{\delta
\hat{f}},\mathbf{A}\cdot\frac{\delta G}{\delta
\bm}\right\}-\left\{\frac{\delta G}{\delta \hat{f}},
\mathbf{A}\cdot\frac{\delta F}{\delta \bm}\right\}\right)\dvol
\end{align}
Then, the standard vector identities
\begin{align*}
&\mathbf{a}\times\mathbf{b}\times\mathbf{c}=\left(\mathbf{a}\cdot\mathbf{c}\right)\mathbf{b}-\left(\mathbf{a}\cdot\mathbf{b}\right)\mathbf{c}
\\
&\nabla\times\left(\mathbf{a}\times\mathbf{b}\right)=\mathbf{a}\left(\nabla\cdot\mathbf{b}\right)-\mathbf{b}\left(\nabla\cdot\mathbf{a}\right)+\left[\mathbf{a},\mathbf{b}\right]
\end{align*}
yield
\begin{align}\nonumber
\{F,G\} =& \int{\bN}\cdot \left[\frac{\delta
F}{\delta{\bm}},\frac{\delta G}{\delta {\bm}}\right]\de^3\bq - \int
\rho \left(\frac{\delta F}{\delta \bm}\cdot\nabla\frac{\delta
G}{\delta \rho}-\frac{\delta G}{\delta \bm}\cdot \nabla\frac{\delta
F}{\delta \rho}\right)\de^3\bq
\\ \nonumber
& - \int \mathbf{A}\cdot \left(\frac{\delta F}{\delta
\bm}\operatorname{div}\frac{\delta G}{\delta
\mathbf{A}}-\frac{\delta G}{\delta
\bm}\operatorname{div}\frac{\delta F}{\delta
\mathbf{A}}-\nabla\times\left(\frac{\delta F}{\delta
\bm}\times\frac{\delta G}{\delta \mathbf{A}}-\frac{\delta G}{\delta
\bm}\times\frac{\delta F}{\delta \mathbf{A}}\right)\right)\de^3\bq
\\ \nonumber
& -q_h \int \mathbf{A}\cdot \left( \nabla n\times\frac{\delta
F}{\delta \bm}\times\frac{\delta G}{\delta \bm} + n\, \nabla\times
\left(\frac{\delta F}{\delta \bm}\times\frac{\delta G}{\delta
\mathbf{m}}\right) -n \left[\frac{\delta F}{\delta \bm},\frac{\delta
G}{\delta \mathbf{m}}\right]\right)\de^3\bq
\\ \nonumber
&+2q_h \int \mathbf{A}\cdot \nabla\times\left(n\,\frac{\delta
F}{\delta \bm}\times\frac{\delta G}{\delta \bm}\right)\de^3\bq +
\int \hat{f}\left\{\frac{\delta F}{\delta \hat{f}},\frac{\delta
G}{\delta \hat{f}}\right\}\dvol
\\ \nonumber
&+q_h \int \hat{f}\left(\left\{\frac{\delta F}{\delta
\hat{f}},\mathbf{A}\cdot\frac{\delta G}{\delta
\bm}\right\}-\left\{\frac{\delta G}{\delta \hat{f}},
\mathbf{A}\cdot\frac{\delta F}{\delta \bm}\right\}\right)\dvol
\end{align}
Upon collecting all vector field commutators and simplifying
according to the general formula
\[
\nabla\times\left(k\,\mathbf{a}\right)=k\,\nabla\times\mathbf{a}+\nabla
k\times\mathbf{a}\,,
\]
we obtain
\begin{align}
\{F,G\} =& \int{\bm}\cdot \left[\frac{\delta
F}{\delta{\bm}},\frac{\delta G}{\delta {\bm}}\right]\de^3\bq - \int
\rho \left(\frac{\delta F}{\delta \bm}\cdot\nabla\frac{\delta
G}{\delta \rho}-\frac{\delta G}{\delta \bm}\cdot \nabla\frac{\delta
F}{\delta \rho}\right)\de^3\bq
\\ \nonumber
& - \int \mathbf{A}\cdot \left(\frac{\delta F}{\delta
\bm}\operatorname{div}\frac{\delta G}{\delta
\mathbf{A}}-\frac{\delta G}{\delta
\bm}\operatorname{div}\frac{\delta F}{\delta
\mathbf{A}}-\nabla\times\left(\frac{\delta F}{\delta
\bm}\times\frac{\delta G}{\delta \mathbf{A}}-\frac{\delta G}{\delta
\bm}\times\frac{\delta F}{\delta \mathbf{A}}\right)\right)\de^3\bq
\\ \nonumber
&+q_h \int n\,\nabla\times\mathbf{A}\cdot \frac{\delta F}{\delta
\bm}\times\frac{\delta G}{\delta \bm}\de^3\bq + \int
\hat{f}\left\{\frac{\delta F}{\delta \hat{f}},\frac{\delta G}{\delta
\hat{f}}\right\}\dvol
\\ \nonumber
&+q_h \int \hat{f}\left(\left\{\frac{\delta F}{\delta
\hat{f}},\mathbf{A}\cdot\frac{\delta G}{\delta
\bm}\right\}-\left\{\frac{\delta G}{\delta \hat{f}},
\mathbf{A}\cdot\frac{\delta F}{\delta \bm}\right\}\right)\dvol
\end{align}
At this point, one expresses the above Poisson bracket in terms of
the distribution
$f(\bq,\bp)=\hat{f}(\bq,\boldsymbol{p}-q_h\mathbf{A})$. Then, the
relations \cite{MaWe1,MaWeRaScSp}
\[
\big\{\hat{h},\hat{k}\big\}=\left\{{h},{k}\right\}+q_h\nabla\times\mathbf{A}\cdot\frac{\partial h}{\partial \bp}\times
\frac{\partial k}{\partial \bp}
\,,\qquad\quad
\frac{\delta F}{\delta \mathbf{A}}=\frac{\delta \bar{F}}{\delta \mathbf{A}}-q_h\, f\frac{\partial}{\partial\bp}\frac{\delta \bar{F}}{\delta f}
\]
and
\begin{align}
\left\{\frac{\delta F}{\delta f},\mathbf{A}\cdot\frac{\delta
G}{\delta \bm}\right\} = -\frac{\partial}{\partial \bp}\frac{\delta
F}{\delta f}\cdot\nabla\left(\frac{\delta G}{\delta
\bm}\cdot\mathbf{A}\right)
\end{align}
transform the bracket to
\begin{align}
\{F,G\} =& \int{\bm}\cdot \left[\frac{\delta
F}{\delta{\bm}},\frac{\delta G}{\delta {\bm}}\right]\de^3\bq - \int
\rho \left(\frac{\delta F}{\delta \bm}\cdot\nabla\frac{\delta
G}{\delta \rho}-\frac{\delta G}{\delta \bm}\cdot \nabla\frac{\delta
F}{\delta \rho}\right)\de^3\bq
\\ \nonumber
& - \int \mathbf{A}\cdot \left(\frac{\delta F}{\delta
\bm}\operatorname{div}\frac{\delta G}{\delta
\mathbf{A}}-\frac{\delta G}{\delta
\bm}\operatorname{div}\frac{\delta F}{\delta
\mathbf{A}}-\nabla\times\left(\frac{\delta F}{\delta
\bm}\times\frac{\delta G}{\delta \mathbf{A}}-\frac{\delta G}{\delta
\bm}\times\frac{\delta F}{\delta \mathbf{A}}\right)\right)\de^3\bq
\\ \nonumber
& -q_h\int f\,\nabla\times\mathbf{A}\cdot\left(\frac{\delta
F}{\delta
\bm}\times\frac{\partial}{\partial\bp}\frac{\delta{G}}{\delta
f}-\frac{\delta G}{\delta
\bm}\times\frac{\partial}{\partial\bp}\frac{\delta{F}}{\delta
f}-\frac{\partial}{\partial\bp}\frac{\delta{F}}{\delta
f}\times\frac{\partial}{\partial\bp}\frac{\delta{G}}{\delta
f}\right)\de^3\bq
\\ \nonumber
&+q_h \int n\,\nabla\times\mathbf{A}\cdot \frac{\delta F}{\delta
\bm}\times\frac{\delta G}{\delta \bm}\de^3\bq +
\int{f}\left\{\frac{\delta F}{\delta{f}},\frac{\delta
G}{\delta{f}}\right\}\dvol
\end{align}
Introducing the magnetic induction $\bB=\nabla\times\mathbf{A}$ by the formula
\[
\frac{\delta
F}{\delta \mathbf{A}}=\nabla\times\frac{\delta
F}{\delta \bB}
\]
completes the proof, thereby yielding the bracket \eqref{PB-current-hybridMHD}.
\,$\blacksquare$

\bigskip

\subsection{Proof of theorem \ref{theorem2}\label{appendix2}}

The momentum shift
$\bM=\bN+\int\!\boldsymbol{p}\,f\,\de^3\boldsymbol{p}$ takes the
Poisson bracket \eqref{PB-MHD+Vlasov} to the form
\begin{align}\nonumber
\{F,G\}=& \int{\bM}\cdot \left[\frac{\delta
F}{\delta{\bM}},\frac{\delta G}{\delta
{\bM}}\right]\de^3\bq
-
\int \rho \left(\frac{\delta F}{\delta \bM}\cdot\nabla\frac{\delta G}{\delta \rho}-\frac{\delta G}{\delta \bM}\cdot \nabla\frac{\delta F}{\delta \rho}\right)\de^3\bq
\\\nonumber
&
-
\int \mathbf{A}\cdot
\left(\pounds_{{\delta F}/{\delta \bM}}\frac{\delta G}{\delta \mathbf{A}}-\pounds_{{\delta G}/{\delta \bM}}\frac{\delta F}{\delta \mathbf{A}}\right)\de^3\bq
\\
&
+
\int \hat{f}\left(\left\{\frac{\delta
F}{\delta \hat{f}},\frac{\delta G}{\delta \hat{f}}\right\}+\left\{\frac{\delta
F}{\delta \hat{f}},\boldsymbol{p}\cdot\frac{\delta G}{\delta \bM}\right\}-\left\{\frac{\delta
G}{\delta \hat{f}},\boldsymbol{p}\cdot\frac{\delta F}{\delta \bM}\right\}\right)\dvol
\end{align}
where $\pounds_{\bu}$ denotes Lie derivative with respect to the velocity ${\bu}$. Then, substitution of the formulas
\[
\big\{\hat{h},\hat{k}\big\}=\left\{{h},{k}\right\}+q_h\nabla\times\mathbf{A}\cdot\left(\frac{\partial h}{\partial \bp}
\times
\frac{\partial k}{\partial \bp}
\right)
\,,\qquad\quad
\frac{\delta F}{\delta \mathbf{A}}=\frac{\delta \bar{F}}{\delta \mathbf{A}}-q_h\int\!f\,\frac{\partial }{\partial\bp}\frac{\delta \bar{F}}{\delta f}\,\de^3\bp
\]
and some integration by parts yield the bracket
\begin{align}\nonumber
\{F,G\}=& \int{\bM}\cdot \left[\frac{\delta
F}{\delta{\bM}},\frac{\delta G}{\delta
{\bM}}\right]\de^3\bq
-
\int \rho \left(\frac{\delta F}{\delta \bM}\cdot\nabla\frac{\delta G}{\delta \rho}-\frac{\delta G}{\delta \bM}\cdot \nabla\frac{\delta F}{\delta \rho}\right)\de^3\bq
\\\nonumber
&
-
\int \mathbf{A}\cdot
\left(\pounds_{{\delta F}/{\delta \bM}\,}\frac{\delta G}{\delta \mathbf{A}}-\pounds_{{\delta G}/{\delta \bM}\,}\frac{\delta F}{\delta \mathbf{A}}\right)\de^3\bq
\\\nonumber
&
+
q_h\int{f}\left(\frac{\partial}{\partial \bp}\frac{\delta
F}{\delta f}\cdot\pounds_{{\delta G}/{\delta \bM}\,}\mathbf{A}-\frac{\partial}{\partial \bp}\frac{\delta
G}{\delta f}\cdot\pounds_{{\delta F}/{\delta \bM}\,}\mathbf{A}\right)\dvol
\\
&
+\int{f}\left(\left\{\frac{\delta
F}{\delta{f}},\frac{\delta G}{\delta{f}}\right\}+q_h\nabla\times\mathbf{A}\cdot\left(\frac{\partial}{\partial \bp}\frac{\delta
F}{\delta f}
\times
\frac{\partial}{\partial \bp}\frac{\delta
G}{\delta f}
\right)
\right)\dvol
\\
&
+
\int{f}\left(\left\{\frac{\delta
F}{\delta \hat{f}},\left(\bp+q_h\mathbf{A}\right)\cdot\frac{\delta G}{\delta \bM}\right\}-\left\{\frac{\delta
G}{\delta \hat{f}},\left(\bp+q_h\mathbf{A}\right)\cdot\frac{\delta F}{\delta \bM}\right\}\right)\dvol
\\
&+
q_h\nabla\times\mathbf{A}\cdot\int{f}\left(\frac{\partial}{\partial \bp}\frac{\delta
F}{\delta f}
\times
\frac{\delta
G}{\delta \bM}
-
\frac{\partial}{\partial \bp}\frac{\delta
G}{\delta f}
\times
\frac{\delta
F}{\delta \bM}
\right)
\dvol
\,.
\end{align}
Noticing that
\begin{align}
\frac{\partial}{\partial \bp}\frac{\delta
F}{\delta f}\cdot\pounds_{{\delta G}/{\delta \bM}\,}\mathbf{A}
=&
\frac{\partial}{\partial \bp}\frac{\delta
F}{\delta f}\cdot\left(\nabla\left(\frac{\delta G}{\delta \bM}\cdot\mathbf{A}\right)-\frac{\delta G}{\delta \bM}\times\nabla\times\mathbf{A}\right)
\\
=&
-\left\{\frac{\delta
F}{\delta f},\mathbf{A}\cdot\frac{\delta G}{\delta \bM}\right\}
-
\nabla\times\mathbf{A}\cdot\frac{\partial}{\partial \bp}\frac{\delta
F}{\delta f}\times\frac{\delta G}{\delta \bM}
\end{align}
takes to the Poisson bracket
\begin{align}\nonumber
\{F,G\}=& \int{\bM}\cdot \left[\frac{\delta
F}{\delta{\bM}},\frac{\delta G}{\delta
{\bM}}\right]\de^3\bq
-
\int \rho \left(\frac{\delta F}{\delta \bM}\cdot\nabla\frac{\delta G}{\delta \rho}-\frac{\delta G}{\delta \bM}\cdot \nabla\frac{\delta F}{\delta \rho}\right)\de^3\bq
\\\nonumber
&
-
\int \mathbf{A}\cdot
\left(\pounds_{{\delta F}/{\delta \bM}\,}\frac{\delta G}{\delta \mathbf{A}}-\pounds_{{\delta G}/{\delta \bM}\,}\frac{\delta F}{\delta \mathbf{A}}\right)\de^3\bq
\\\nonumber
&
+
\int f\left\{\frac{\delta
F}{\delta f},\frac{\delta G}{\delta f}\right\}\dvol
+q_h\nabla\times\mathbf{A}\cdot
\int \!f\left(\frac{\partial}{\partial \bp}\frac{\delta
F}{\delta f}
\times
\frac{\partial}{\partial \bp}\frac{\delta
G}{\delta f}
\right)\dvol
\\
&+\int f\left(\left\{\frac{\delta
F}{\delta f},\bp\cdot\frac{\delta G}{\delta \bM}\right\}-\left\{\frac{\delta
G}{\delta f},\bp\cdot\frac{\delta F}{\delta \bM}\right\}\right)\dvol
\,.
\end{align}
Then, using the formula
\[
\frac{\delta
F}{\delta \mathbf{A}}=\nabla\times\frac{\delta
F}{\delta \bB}
\]
yields the Poisson structure \eqref{PB-pressure-hybridMHD}. $\blacksquare$


\newpage

\begin{thebibliography}{99}

\bibitem{BeDeCh}
Belova E.V.; Denton R.E.; Chan A.A. {\it Hybrid simulations of the
effects of energetic particles on low-frequency MHD waves.} J.
Comput. Phys. 136 (1997), no. 2, 324--336

\bibitem{BrHa}
Brizard, A.J.; Hahm, T.S. {\it Foundations of nonlinear gyrokinetic
theory.} Rev. Mod. Phys. 79 (2007), no. 2,  421--468

\bibitem{BrTr}
Brizard, A.J.; Tracy, E.R. {\it Mini-conference on Hamiltonian and Lagrangian methods in fluid and plasma physics.} Phys. Plasmas 10 (2003), no. 5, 2163--2168

\bibitem{ByCoCoHa}
Byers, J.A.; Cohen, B.I.; Condit, W.C.; Hanson, J.D. {\it Hybrid
simulations of quasineutral phenomena in magnetized plasma.} J.
Comp. Phys. 27 (1978), no. 3, 363--396

\bibitem{CeMaHo}
Cendra, H.; Holm, D.D.; Hoyle, M.J.W.; Marsden, J.E. {\it The
Maxwell-Vlasov equations in Euler-Poincaré form.}  J. Math. Phys. 39
(1998),  no. 6, 3138--3157

\bibitem{ChFr}
Chandy, A.J.; Frankel, S.H.  {\it Regularization-based sub-grid scale (SGS) models for large eddy
simulations (LES) of high-$Re$ decaying isotropic turbulence.} J. Turbul. 10 (2009), no. 25, 1--22

\bibitem{Cheng}
Cheng, C.Z. {\it A kinetic-magnetohydrodynamic model for
low-frequency phenomena.} J. Geophys. Res. 96 (1991), no. A12,
21,159--21,171

\bibitem{ChengJohnson}
Cheng, C.Z.; Johnson, J.R. {\it A kinetic-fluid model.} J. Geophys.
Res. 104 (1999), no. A1, 413--427.

\bibitem{Fla}
Fl{\aa}, T. {\it A hybrid fluid-kinetic theory for plasma physics.}
{\tt arXiv:math-ph/0110036v1}

\bibitem{FoHoTi}
Foias, C.; Holm, D.D.; Titi,  E.S. {\it The Navier-Stokes-alpha
model of fluid turbulence.} Phys. D 152/3 (2001) 505–519

\bibitem{FuPark}
Fu, G.Y.; Park, W. {\it Nonlinear hybrid simulation of the
toroidicity-induced Alfv\'en Eigenmode.} Phys. Rev. Lett. 74 (1995),
no. 9, 1594--1596

\bibitem{GiHoTr}
Gibbons, J.; Holm, D.D.; Tronci, C. {\it Geometry of Vlasov kinetic moments: a bosonic Fock space for the symmetric Schouten bracket.} Phys. Lett. A 372 (2008), 4184--4196

\bibitem{GrKaLi}
Grebogi, C.; Kaufman, A.N.; Littlejohn, R.G. {\it Hamiltonian theory
of ponderomotive effects of an electromagnetic wave in a nonuniform
magnetic field.}  Phys. Rev. Lett.  43  (1979), no. 22, 1668--1671.

\bibitem{Holm1}
Holm, D.D. {\it Hamiltonian dynamics of a charged fluid, including
electro- and magnetohydrodynamics.}  Phys. Lett. A  114 (1986),  no.
3, 137--141.

\bibitem{Holm2}
Holm, D.D. {\it Hamiltonian dynamics and stability analysis of
neutral electromagnetic fluids with induction.}  Phys. D  25
(1987), no. 1-3, 261--287.

\bibitem{Holm3}
Holm, D.D. {\it Hall magnetohydrodynamics, conservation laws and
Lyapunov stability.} Phys Fluids 30 (1987), no. 5, 1310--1322.

\bibitem{HolmKupershmidt}
Holm, D.D.; Kupershmidt, B.A. {\it Noncanonical Hamiltonian
formulation of ideal magnetohydrodynamics.} Phys. D  7  (1983), no.
1-3, 330--333.

\bibitem{HoMaRaWe}
Holm, D.D.; Marsden, J.E.; Ratiu, T.S.; Weinstein, A. {\it Nonlinear
stability of fluid and plasma equilibria.}  Phys. Rep.  123 (1985),
no. 1-2, 1--116

\bibitem{HoTr}
Holm, D.D.; Tronci, C. {\it Geodesic Vlasov equations and their integrable moment closures.} J. Geom. Mech. 1 (2009), no. 2, 181--288.

\bibitem{HoTr2010}
Holm, D.D.; Tronci, C. {\it Euler-Poincar\'e formulation of hybrid
plasma models.} Preprint

\bibitem{KimEtAl}
Kim, C.C.; Sovinec, C.R.; Parker, S.E.; the NIMROD team {\it Hybrid kinetic-MHD simulations in general geometry.} Comp. Phys. Comm. 164 (2004), 448--455

\bibitem{Liboff}
Liboff, R.L.; Kinetic Theory, Springer-Verlag 1998.

\bibitem{Littlejohn}
Littlejohn, R.G. {\it Variational principles of guiding centre
motion.} J. Plasma Phys. 29 (1983), no. 1,  111--125

\bibitem{Low}
Low, F.E. {\it A Lagrangian formulation of the Boltzmann-Vlasov equation for plasmas.} Proc. R. Soc. London, Ser. A  248 (1958), 282--287

\bibitem{MaWe1}
Marsden, J.E.; Weinstein, A. {\it The Hamiltonian structure of the
Maxwell-Vlasov equations.}  Phys. D  4  (1981/82), no. 3, 394--406

\bibitem{MaWeRa}
Marsden, J.E.; Weinstein, A.; Ratiu, T.S. {\it Semidirect products
and reduction in mechanics.}  Trans. Amer. Math. Soc.  281  (1984),
no. 1, 147--177.

\bibitem{MaWeRaScSp}
Marsden, J.E.; Weinstein, A.; Ratiu, T.; Schimd, R.; Spencer, R.G.
{\it  Hamiltonian systems with symmetry, coadjoint orbits and plasma
physics},  Atti Accad. Sci. Torino Cl. Sci. Fis. Mat. Natur.  117
(1983), no. 1, 289--340

\bibitem{Morrison2005}
Morrison, P.J. {\it Hamiltonian and action principle formulations of plasma physics.} Phys. Plasmas 12 (2005), 058102

\bibitem{Morrison1bis}
Morrison, P.J. {\it Hamiltonian field description of the
one-dimensional Poisson-Vlasov system.} Princeton Plasma Physics
Laboratory Report, PPPL-1788 (1981).

\bibitem{Morrison1}
Morrison, P.J. {\it Hamiltonian field description of two-dimensional
vortex fluids and guiding center plasmas.} Princeton
Plasma Physics Laboratory Report, PPPL-1783 (1981).

\bibitem{Morrison2}
Morrison, P.J. {\it On Hamiltonian and action principle
formulations of plasma dynamics.} AIP Conf. Proc., 1188 (2009),
329--344.

\bibitem{Morrison3}
Morrison, P.J. {\it Poisson brackets for fluids and plasmas.} AIP Conf. Proc. 88 (1982), 13--46

\bibitem{Morrison2bis}
Morrison, P.J. {\it The Maxwell-Vlasov equations as a continuous
Hamiltonian system.} Phys. Lett. A 80 (1986), no. 5--6, 383--386

\bibitem{MorrisonGreene}
Morrison, P.J.; Greene, J.M. {\it Noncanonical Hamiltonian density
formulation of hydrodynamics and ideal magnetohydrodynamics.} Phys.
Rev. Lett. 45 (1980), 790--794.

\bibitem{PaBeFuTaStSu}
Park, W.; Belova, E.V.; Fu, G.Y.; Tang, X.Z.; Strauss, H.R.;
Sugiyama, L.E. {\it Plasma simulation studies using multilevel
physics models.} Phys. Plasmas 6 (1999), no. 6, 1796--1803.

\bibitem{ParkEtAl}
Park, W.; Parker, S.; Biglari, H.; Chance, M.; Chen, L.; Cheng,
C.Z.; Hahm, T.S.; Lee, W.W.; Kulsrud, R.; Monticello, D.; Sugiyama,
L.; White, R. {\it Three-dimensional hybrid
gyrokinetic-magnetohydrodynamics simulation.} Phys. Fluids B 4
(1992), no. 7, 2033--2037

\bibitem{Pfirsch}
Pfirsch, D. {\it New variational formulation of Maxwell-Vlasov and
guiding center theories local charge and energy conservation
 laws.} Z. Naturforsch. A 39a (1984), 1--8

\bibitem{PfMo1}
Pfirsch, D.; Morrison, P.J. {\it Local conservation laws for the
Maxwell-Vlasov and collisionless kinetic guiding-center theories.}
Phys. Rev. A 3 (1985), no. 2, 1714-–1721

\bibitem{PfMo2}
Pfirsch, D.; Morrison, P.J. {\it The energy-momentum tensor for the
linearized Maxwell–Vlasov and kinetic guiding center theories.}
Phys. Fluids B 3 (1991), no. 2, 271--283


\bibitem{Spencer}
Spencer, R.G. {\it The Hamiltonian structure of multi-species fluid
electrodynamics.} AIP Conf. Proc., 88 (1982), 121--126.

\bibitem{SpKa}
Spencer, R.G.; Kaufman, A.N. {\it Hamiltonian structure of two-fluid
plasma dynamics.}  Phys. Rev. A (3)  25  (1982), no. 4, 2437--2439

\bibitem{TaBrKi}
Takahashi, R.;  Brennan, D. P.; Kim, C.C. {\it Kinetic effects of
energetic particles on resistive MHD stability.} Phys. Rev. Lett.
102 (2009), 135001

\bibitem{ToSaWaWaHo}
Todo, Y.; Sato, T.; Watanabe, K.; Watanabe, T. H.; Horiuchi, R. {\it
Magnetohydrodynamic Vlasov simulation of the toroidal Alfv\'en
eigenmode.} Phys. Plasmas 2 (1995), no. 7, 2711--2716.

\bibitem{Weinstein}
Weinstein, A. {\it Gauge groups and Poisson brackets for interacting
particles and fields.} AIP Conf. Proc., 88 (1982),  1--11.

\bibitem{WiYiOmKaQu} Winske, D.; Yin, L.; Omidi, N.; Karimbadi, H.; Quest, K. {\it
Hybrid simulation codes: past, present and future -- a tutorial.}
Lect. Notes Phys. 615 (2003), 136--165

\bibitem{YeMo}
Ye, H.; Morrison, P.J. {\it Action principles for the Vlasov
equation.} Phys. Fluids B 4 (1992), no. 4, 771--777

\bibitem{YiWiEtAl}
Yin, L.; Winske, D.; Gary, S.P.; Birn, J. {\it Hybrid and Hall-magnetohydrodynamics simulations of collisionless reconnection: effect of the ion pressure tensor and particle Hall-magnetohydrodynamics.} Phys. Plasmas 9 (2002), no. 6, 2575--2584

\end{thebibliography}
\end{document}